\documentclass{article}
\usepackage[latin1]{inputenc}
\usepackage[english]{babel}
\usepackage{dsfont}
\usepackage[cyr]{aeguill}
\usepackage{amsmath,amssymb}
\usepackage{amsthm}
\usepackage{eurosym}
\usepackage{multicol}
\usepackage{color}
\usepackage{fancyhdr}
\usepackage[pdftex]{graphicx}
\usepackage[pdftex,colorlinks=automatic,linkcolor=black,citecolor=black,urlcolor=black]{hyperref}
\usepackage{verbatim}
\usepackage{stmaryrd}

\newcommand{\proba}{\mathbb{P}} 
\newcommand{\indic}{\mathds{1}}
\newcommand{\E}{\mathbb{E}} 
\newcommand{\Cov}{\text{Cov}} 
\newcommand{\Var}{\text{Var}} 

\theoremstyle{plain}
\newtheorem{thm}{Theorem}
\newtheorem{pro}{Proposition}

\newtheorem{lem}{Lemma}
\newtheorem{defn}{Definition}

\providecommand{\keywords}[1]{\textbf{\textit{Keywords --}} #1}

\title{Forecasting with fractional Brownian motion: \\ a financial perspective}

\author{Matthieu Garcin\footnote{ Léonard de Vinci Pôle Universitaire, Research center, 92916 Paris La Défense, France, matthieu.garcin@m4x.org. The author is grateful to Martino Grasselli, Chafic Merhy, and Konstantinos Tsianos for valuable discussions and comments.}}

\date{\today}

\begin{document}

\maketitle

\begin{abstract}
The fractional Brownian motion (fBm) extends the standard Brownian motion by introducing some dependence between non-overlapping increments. Consequently, if one considers for example that log-prices follow an fBm, one can exploit the non-Markovian nature of the fBm to forecast future states of the process and make statistical arbitrages. We provide new insights into forecasting an fBm, by proposing theoretical formulas for accuracy metrics relevant to a systematic trader, from the hit ratio to the expected gain and risk of a simple strategy. In addition, we answer some key questions about optimizing trading strategies in the fBm framework: Which lagged increments of the fBm, observed in discrete time, are to be considered? If the predicted increment is close to zero, up to which threshold is it more profitable not to invest? We also propose empirical applications on high-frequency FX rates, as well as on realized volatility series, exploring the rough volatility concept in a forecasting perspective.
\end{abstract}


\section{Introduction}

An fBm is a non-Markovian process. A question thus naturally arises when considering the use of this model in finance for describing the dynamic of prices or log-prices: Does the fBm induce pure arbitrages or not? This question has led to numerous articles~\cite{BSV}. In particular, it has been shown that arbitrage opportunities exist when trading in continuous time is allowed~\cite{Rogers}. But the answer is different when one trades in discrete time, even at a very high frequency~\cite{Cheridito}. From the perspective of a financial market practitioner, beyond the strong mathematical interest of the question, the distinction is not so clear between continuous time and discrete time with arbitrary high frequency. Therefore, like others before, we believe that this debate about the existence of arbitrages implied by this model is mainly motivated by theoretical convenience, what is not enough to exclude the fBm as a model of log-prices~\cite{Cont}. 

A more recent literature deals with the use of fBm in statistical arbitrage~\cite{GNR,GMR}. Exploiting the autocovariance structure of the fBm indeed makes it possible to forecast future states of the process~\cite{NP}. This simple idea has been used since in algorithmic trading to build systematic strategies. Such strategies do not lead to a certain gain as do the pure arbitrages previously mentioned, but they are profitable in average. All the contributions on this subject highlight the asymmetry of the performance of these strategies with respect to $H=1/2$, where $H$ is the Hurst exponent. 

Asymmetry with respect to $H=1/2$ is also underlined in the econophysics literature. Indeed, when $H>1/2$, the increments of the fBm are positively correlated and the process also has a long memory. By contrast, when $H<1/2$, non-overlapping increments are negatively correlated and there is no long memory, in the sense that the autocovariance function decreases exponentially. Many researchers equate the notion of market efficiency with the property of long memory, leading them to consider that the market is efficient if and only if $H>1/2$~\cite{DiMatteo2003,CT,DiMatteo2005,ARAR,ECOJ,GSP}. We indeed aknowledge that when $H<1/2$, the performance of the fBm-based predictor does not always lead to satisfying empirical results~\cite{Garcin2017}. But this practical limitation does not come from a particular property of the fBm itself. It comes instead from the fact that the model is not well specified or from the related difficulties to estimate the parameters properly, in a nutshell from model risk~\cite{GarcinLamperti,GarcinEstimLamp}. Besides the identification of market efficiency to long-range dependence, another branch of the econophysics literature considers that the market is efficient for $H=1/2$ and that its inefficiency gradually increases as $H$ gets away from $1/2$~\cite{KV13,KV16,Kristoufek,BP,AG}. Indeed, predictions do not rely on the long-range dependence property but, instead, on the autocovariance of the process, which is different from zero as soon as $H\neq 1/2$.

While the fBm is widespread in econophysics and more recently in mainstream quantitative finance, not everything has been said yet on forecasting the fBm in a financial perspective. First, the theoretical evaluation of the quality of the forecast invoked in the literature is rarely very appropriate for systematic traders. For example, the mean squared error (MSE), put forward in several papers~\cite{NP,GJR}, albeit statistically relevant, is not related to the performance and risk of a trading strategy. In this perspective, other metrics have to be applied in the fBm setting. Moreover, while many articles focus on the fBm in continuous time~\cite{GNR,GMR}, the reality of trading, which is affected by liquidity frictions, is in discrete time, whether for the observation of the price process or for the instants at which one is able to trade. And this limitation has some consequences. For example, determining which lagged price returns should be used as input of the predictor is overriding in order to optimize an fBm-based forecast and related systematic trading strategies. 

We answer these questions in the present paper. We introduce several accuracy metrics for covariance-based predictors: a hit ratio, an average gain, a risk defined as a lower semi-deviation, and the resulting risk-adjusted performance. The risk measure we have chosen differentiates our work from the traditional mean-variance framework~\cite{GNR,GMR}. While the variance is a widespread deviation risk measure, we indeed think that a lower semi-deviation more appropriately depicts downside risk. We provide theoretical expressions for all the aforementioned metrics, using the fBm assumption. All these expressions depict the link between the Hurst exponent $H$ and the forecasting ability of the fBm: the closer $H$ is to $1/2$, the worst is the quality of the forecast, and, asymmetrically, the forecast is better for $H>1/2$ than for $1-H$. We also use the theoretical expressions mentioned above to build tools which are useful in optimizing trading strategies. We do this in two directions. First, we find numerically a general expression for the duration of the optimal time lags of price returns used as input variables of the predictor. Second, we define an optimal threshold under which we consider that the forecast is too close to zero for an investment to positively impact the risk-adjusted profit. This last idea finally leads to reducing the trading frequency and thus to take into account liquidity constraints.

We also propose two empirical applications highlighting the relevance of predictions based on the fBm in finance. In particular, we maximize an ex-ante risk-adjusted performance of a trading strategy on FX rates, where log-prices follow an fBm. We also investigate another application of the fBm in finance than the sole log-price process and for which a performing forecast method is useful for trading desks. We indeed study series of realized volatilities, exploring the rough volatility and how one can optimize the forecast of volatility in this context, focusing on hit ratios.

The rest of the paper is organized as follows. In Section~\ref{sec:prelim}, we recall some basic and useful properties of the fBm. Section~\ref{sec:forecast} introduces the art of forecasting an fBm, namely the covariance-based predictor, its accuracy measured by a hit ratio, and the optimal selection of its input. In Section~\ref{sec:statarb}, we focus on statistical arbitrage, providing both a simple systematic trading strategy and related accuracy metrics. In Section~\ref{sec:empiric}, we present empirical results on high-frequency FX rates and on daily financial series, including realized volatility. Section~\ref{sec:conclu} concludes.

\section{Model description: preliminaries on the fBm}\label{sec:prelim}

The fBm, introduced by Mandelbort and van Ness in 1968, can follow several equivalent definitions~\cite{MvN}. Among them, we can cite the integral-based definition.

\begin{defn}\label{def:fBmInt}
An fBm of Hurst exponent $H\in(0,1)$ and volatility parameter $\sigma>0$ is a stochastic process $X_t$ such that, $\forall t\in\mathbb R$,
$$X_t=\frac{\sigma}{\Gamma\left(H+\frac{1}{2}\right)}\int_{-\infty}^{\infty}{\left((t-s)_+^{H-1/2}-(-s)_+^{H-1/2}\right)dW_s},$$
where $W_s$ is a standard Brownian motion.
\end{defn}

In Definition~\ref{def:fBmInt}, the fBm is to be seen as a weighted average of a Gaussian white noise. The weights are defined by the kernel $(s,t)\mapsto\left((t-s)_+^{H-1/2}-(-s)_+^{H-1/2}\right)$, whose right part, $-(-s)_+^{H-1/2}$, is simply intended to make the process $X_t$ equal to zero when $t=0$, which is a feature also stated in other definitions of the fBm. One major motivation for introducing the fBm was to define a process with other scaling properties than the standard Brownian motion. This alternative scaling clearly appears in Definition~\ref{def:fBmVar}, in the expression of the variance of an increment of the process $X_t$, but it is not so obvious in the integral-based definition. The idea of Mandelbrot and van Ness to obtain a specific scaling was to use fractional calculus. Indeed, Definition~\ref{def:fBmInt} corresponds to the fractional derivative (respectively integral) of order $1/2-H$ (resp. $H-1/2$) of a standard Brownian motion, if $H<1/2$ (resp. $H>1/2$), the case $H=1/2$ corresponding to the standard Brownian motion itself. The scaling feature characterized by a given Hurst exponent is thus to be related to a specific autocovariance of the increments of the process in the fBm model.

Since we will work in discrete time in all the paper, we introduce another definition of the fBm, which is equivalent to Definition~\ref{def:fBmInt}.

\begin{defn}\label{def:fBmVar}
An fBm of Hurst exponent $H\in(0,1)$ and volatility parameter $\sigma>0$ is a stochastic Gaussian process $X_t$ such that $X_0=0$, $\E\{X_t\}=0$, and 
$$\E\{(X_t-X_s)^2\}=\sigma^2|t-s|^{2H}$$
for all $s,t\in\mathbb R$.
\end{defn}

Thanks to Definition~\ref{def:fBmVar}, a particular property of scaling, or selfsimilarity, clearly appears for the fBm. This is the purpose of Proposition~\ref{pro:selfsim}.

\begin{pro}\label{pro:selfsim}
Let $X_t$ be an fBm of Hurst exponent $H\in(0,1)$. Then, the process $X_t$ is statistically $H$-selfsimilar, that is to say, $\forall c>0$, the random variables $X_{ct}$ and $c^H X_{t}$, for a given $t\in\mathbb R$, follow the same distribution.
\end{pro} 

\begin{proof}
According to Definition~\ref{def:fBmVar}, the two variables $X_{ct}$ and $c^H X_{t}$ are Gaussian, have both a mean equal to zero and a variance equal to $\sigma^2c^{2H}|t|^{2H}$. They thus follow exactly the same distribution. 
\end{proof}

This scaling property of the fBm also leads to a specific autocovariance of the process and of its increments, as exposed in Proposition~\ref{pro:cov_fBm}.

\begin{pro}\label{pro:cov_fBm}
Let $X_t$ be an fBm of Hurst exponent $H\in(0,1)$ and volatility parameter $\sigma>0$. Then, the covariance of $X_t$ and $X_s$ is
\begin{equation}\label{eq:covMBF}
\E\{X_tX_s\}=\frac{\sigma^2}{2}(|t|^{2H}+|s|^{2H}-|t-s|^{2H})
\end{equation}
and the covariance between increments $X_t-X_s$ and $X_v-X_u$ is
\begin{equation}\label{eq:covMBF_incr}
\E\{(X_t-X_s)(X_v-X_u)\}=\frac{\sigma^2}{2}(|u-t|^{2H}+|v-s|^{2H}-|v-t|^{2H}-|u-s|^{2H}).
\end{equation}
\end{pro}

\begin{proof}
Using Definition~\ref{def:fBmVar} and in particular its consequence that $\E\{X_t^2\}=\E\{(X_t-X_0)^2\}=\sigma^2|t|^{2H}$, the proof of Proposition\ref{pro:cov_fBm} is straightforward:
$$\begin{array}{ccl}
\E\{X_tX_s\} & = & -\frac{1}{2}\left(\E\{(X_t-X_s)^2\}-\E\{X_t^2\}-\E\{X_s^2\}\right) \\
 & = & -\frac{1}{2}\left(\sigma^2|t-s|^{2H}-\sigma^2|t|^{2H}-\sigma^2|s|^{2H}\right) \\
 & = & \frac{\sigma^2}{2}(|t|^{2H}+|s|^{2H}-|t-s|^{2H})
\end{array}$$
and
$$\begin{array}{ccl}
\E\{(X_t-X_s)(X_v-X_u)\} & = & \E\{X_tX_v\}-\E\{X_tX_u\}-\E\{X_sX_v\}+\E\{X_sX_u\} \\
& = & \frac{\sigma^2}{2}(|u-t|^{2H}+|v-s|^{2H}-|v-t|^{2H}-|u-s|^{2H}).
\end{array}$$
\end{proof}

From Proposition~\ref{pro:cov_fBm}, one sees that the fBm is not a stationary process but its increments are stationary, exactly like the standard Brownian motion. This means that, for making forecasts, we will more easily work with increments than with the process itself.

In all what follows, the process $X_t$ is an fBm of Hurst exponent $H$ and volatility parameter $\sigma>0$. In the empirical part of the paper, $X_t$ will either depict a log-price or the volatility process of a price. We thus cover the two main applications of the fBm in finance. The first one has been put forward by econophysicists and the second one has known a recent interest in mathematical finance with the rise of the literature about rough volatility. For simplification, we will sometimes evoke log-prices in the theoretical developments, but the approaches can also often be applied to volatilities and to any other dynamic described by an fBm, in particular in Section~\ref{sec:forecast}. Besides, even though volatility is not tradable, its forecast value is used in many trading algorithms, making the forecast of rough volatility a relevant challenge in finance. However, Section~\ref{sec:seuilOptim}, in which we determine the profit and risk of a systematic strategy, is only appropriate for log-prices.

\section{Forecasting an fBm}\label{sec:forecast}

In this section, we present the standard formula for the forecast of the process in discrete time along with various accuracy metrics.

\subsection{Covariance-based forecast}

We want to predict the value of the process $X_t$ for a time horizon $h$, that is $X_{t+h}$, conditionally to observations at some times $I=\{t_1,...,t_n\}$, such that $\forall s\in I,s\leq t$, where $t$ is the current time. Writing $Y=(X_{t_1} \ ... \ X_{t_n})^T$, the forecast minimizing the MSE is $\hat X_{t+h|Y}=\E\{X_{t+h}|Y\}$. We recall that $X_t$ is an fBm. For this reason, we can easily determine the covariance between $X_s$ and $X_t$ at any times $s,t\in\mathbb R$, thanks to equation~\eqref{eq:covMBF}, and we can even build covariance matrices for $X_{t+h}$ and the vector $Y$. This makes it possible to express explicitly $\hat X_{t+h|Y}$ as well as its MSE. This is the purpose of Proposition~\ref{pro:NP}, for which we do not detail the proof since it is a direct consequence of the Gaussian conditioning theorem~\cite{NP}.

\begin{pro}\label{pro:NP}
Let $X_t$ be an fBm, $h>0$, $Y=[X_{t_1} \ ... \ X_{t_n}]^T$, with $\forall i<j$, $t_i\neq 0$ and $t_i<t_{j}$. The estimator of $X_{t+h}$ minimizing the MSE conditionally to $Y$ is
$$\hat X_{t+h|Y}=\E\{X_{t+h}|Y\}=\Sigma_{XY}\Sigma_Y^{-1}Y$$
and the corresponding MSE is
$$\E\{(\hat X_{t+h|Y}-X_{t+h})^2\}=\Sigma_X-\Sigma_{XY}\Sigma_Y^{-1}\Sigma_{XY}^T,$$
where $\Sigma_Y=\E\{YY^T\}$, $\Sigma_X=\E\{X_{t+h}^2\}$, and $\Sigma_{XY}=\E\{X_{t+h}Y^T\}$ can be explicitly expressed thanks to equation~\eqref{eq:covMBF}.
\end{pro}

We also stress the fact that using existing results about forecasting an fBm in continuous time in our framework is less relevant. Indeed, the adaptation of the continuous-time formulas to the discrete-time framework needs the discretization of an integral and therefore only leads to an approximation of $\E\{X_{t+h}|Y\}$. On the contrary, exploiting the covariance matrix of an fBm is both easier and more accurate.

It is also well known that, with the fBm model, the prediction at time horizon $h$ is optimal when the only past observation considered has a time lag also equal to $h$. In particular, the optimality can be understood as the minimization of the MSE. In other words, for $n=2$, considering $t_1=t$ and $t_2=t-h$ is the choice of dates which minimizes the MSE in Proposition~\ref{pro:NP}.

\subsection{Hit ratio}

The MSE, while being a useful statistical tool, may be less relevant in finance than other statistics, such as the hit ratio, which is the probability to make a good forecast of the sign of a future price return. Indeed, in order to make an investment decision, a trader forecasts at time $t$ the price return $R_{t,t+h}$ between $t$ and $t+h$, defined by
\begin{equation}\label{eq:incr}
R_{t,t+h}=X_{t+h}-X_t,
\end{equation}
by a specific predictor $\hat R_{t,t+h}$. If $\hat R_{t,t+h}>0$ (respectively $<0$), the trader must be long (resp. short) between $t$ and $t+h$ in order to expect a gain at time $t+h$. Therefore, the hit ratio is implicitly related to the performance of a trading strategy. Every trader is thus able to interpret the value of a hit ratio and to discard any forecast with a hit ratio lower than $50\%$. On the contrary, the interpretation of the MSE is not as clear and this metric is only useful for comparing several forecasting methods on the same dataset, not for judging the intrinsic quality of a given predictor.

In our framework, log-prices follow an fBm. Therefore, log-returns are simply increments of the fBm. This will ease the calculation and the analysis of the hit ratio, since increments of the fBm are stationary whereas the fBm itself is not stationary. Moreover, the forecast is a weighted sum of past observations. If one considers prices instead of returns, the weights will depend on $t$ and not only on the time lags. This is striking for example if $t=h$. Indeed, as exposed above, when considering two observations for forecasting, the optimal time lag in the past should be $h$ and thus one should forecast $X_{t+h}$ using $X_t$ and $X_{t-h}=X_{0}=0$. But, in this case, the matrix $\Sigma_Y$ is not invertible and we cannot apply Proposition~\ref{pro:NP}.

Before determining the hit ratio, we need to detail the predictor of the true future log-return $R_{t,t+h}$. We forecast $R_{t,t+h}$ conditionally to a vector $S$ of adjacent past price returns. The vector $S$ is built with the help of the set of time lags $\Delta=\{\delta_0,...,\delta_n\}$, with $n\geq 1$ and $\forall i<j,\delta_i<\delta_j$ and $\delta_0\geq 0$. More precisely, we have
$$S=[R_{t-\delta_1,t-\delta_0}\ ...\ R_{t-\delta_n,t-\delta_{n-1}}]^T.$$
For convenience, we set $\delta_0=0$, because it is natural to consider the current state in our forecast, but for the theoretical results exposed below, this constraint is not mandatory. We note our predictor $\hat R_{t,t+h|\Delta}$. The extension of Proposition~\ref{pro:NP} to the increments of the process instead of the process itself is straightforward, since it simply consists in replacing the matrices of covariance of the fBm by the matrices of covariance of its increments.

\begin{pro}\label{pro:NP_incr}
Let $X_t$ be an fBm of Hurst exponent $H\in(0,1)$ and volatility parameter $\sigma>0$, $h>0$, $S=[R_{t-\delta_1,t-\delta_0}\ ...\ R_{t-\delta_n,t-\delta_{n-1}}]^T$, with $\forall i<j$, $\delta_i<\delta_{j}$, $\delta_i\geq 0$, and $R_{.,.}$ expressed by equation~\eqref{eq:incr}. The estimator of $R_{t,t+h}$ minimizing the MSE conditionally to $S$ is
\begin{equation}\label{eq:r_hat}
\hat R_{t,t+h|\Delta}=\Sigma_{RS}\Sigma_{S}^{-1}S
\end{equation}
and the corresponding MSE is
$$\E\{(\hat R_{t,t+h|\Delta}-R_{t,t+h})^2\}=\Sigma_R-\Sigma_{RS}\Sigma_S^{-1}\Sigma_{RS}^T,$$
where $\Sigma_S=\E\{SS^T\}$, $\Sigma_R=\E\{R_{t,t+h}^2\}=\sigma^2h^{2H}$, and $\Sigma_{RS}=\E\{R_{t,t+h}S^T\}$ can be explicitly expressed thanks to equation~\eqref{eq:covMBF_incr}.
\end{pro}

It is worth noting that $\hat R_{t,t+h|\Delta}$ does only depend on $h$, not on $t$. In what follows, we keep $t$ in most subscripts, but the hit ratios, risk measures, and expected performance displayed in the coming theorems and propositions in fact only depend on the forecast horizon $h$.

One can write equation~\eqref{eq:r_hat} differently, as a weighted sum of past returns:
$$\hat R_{t,t+h|\Delta}=\sum_{i=1}^{n}{\beta_iR_{t-\delta_i,t-\delta_{i-1}}},$$
where the $\beta_i$ must be consistent with the covariance approach exposed in Proposition~\ref{pro:NP_incr}: 
$$[\beta_1\ ...\ \beta_n]=\Sigma_{RS}\Sigma_{S}^{-1}.$$
In the particular case where $n=1$, that is when the predicted return only relies on one past return, namely $\hat R_{t,t+h|\Delta}=\beta_1R_{t-\delta_1,t}$, we have the following expression for $\beta_1$:
$$\beta_1=\frac{\Cov(R_{t,t+h},R_{t-\delta_1,t})}{\Var(R_{t-\delta_1,t})},$$
according to Proposition~\ref{pro:NP_incr}. We note that the covariance and variance appearing in the equation above are not estimated with their empirical version. The $\beta_1$ parameter is indeed model-dependent. The fBm assumption will in particular make it possible to rely on a parsimonious representation of the dependence structure among price returns, even when considering a big number of time lags, and thus to avoid overfitting in the forecast. This specific feature of parsimony will be studied empirically in Section~\ref{sec:rough}. Taking into account the fBm assumption, the expression for $\beta_1$, when $n=1$ is:
\begin{equation}\label{eq:beta1}
\beta_1=\frac{1}{2}\left[\left(\frac{h}{\delta_1}+1\right)^{2H}-\left(\frac{h}{\delta_1}\right)^{2H}-1\right],
\end{equation}
according to equation~\eqref{eq:covMBF_incr}. It is worth noting that $\beta_1$ is positive when $H>1/2$, equal to zero for $H=1/2$, and negative when $H<1/2$. It is consistent with properties of the fBm: increments are positively correlated when $H>1/2$, non correlated when $H=1/2$, and negatively correlated when $H<1/2$. In particular, we note that when $\delta_1=h$, $\beta_1$ tends toward 1 when $H$ tends toward 1: the persistence is such that the best prediction of the future return is the past return of same duration. Alternatively, in the anti-persistence case, when $H$ tends toward 0, whatever the value of $\delta_1$, $\beta_1$ tends toward $-1/2$.

In the general case, that is for $n$ not necessarily equal to 1, we want to determine the theoretical hit ratio of our predictor. As written above, it is the probability to forecast properly the sign of the future price return $R_{t,t+h}$. Two definitions are in fact possible, depending on the conditioning of the probability to past returns or not.

\begin{defn}\label{def:hit}
Let $\hat R_{t,t+h|\Delta}$ be a predictor of $R_{t,t+h}$ based on the vector $S$ of past returns, as defined by equation~\eqref{eq:r_hat}. The conditional hit ratio of this predictor is
$$\rho^c(y)=\proba\left[\hat R_{t,t+h|\Delta}R_{t,t+h}\geq 0\left| S=y\right.\right]$$
and the non-conditional hit ratio is
$$\rho=\proba\left[\hat R_{t,t+h|\Delta}R_{t,t+h}\geq 0\right].$$
\end{defn}

In the fBm framework, the vector $S$ admits a Gaussian density $g_S$ of mean zero and variance $\Sigma_S$. The non-conditional hit ratio is the weighted average of the conditional hit ratio:
$$\rho=\int_{\mathbb R}{\rho^c(y)g_S(y)dy}=\E\left\{\rho^c(S)\right\}.$$

Theorem~\ref{th:hitNrdt} provides a theoretical expression both for $\rho^c$ and $\rho$, when the price returns are assumed to be increments of an fBm.

\begin{thm}\label{th:hitNrdt}
Let $X_t$ be an fBm of Hurst exponent $H\in(0,1)$ and volatility parameter $\sigma>0$. Let $h>0$ and $\hat R_{t,t+h|\Delta}$ be a predictor of $R_{t,t+h}$ based on the vector $S$ of past returns, as defined by equation~\eqref{eq:r_hat}. The conditional hit ratio is
$$\rho^c(y)= N\left(\frac{\left|\Sigma_{RS}\Sigma_S^{-1}y\right|}{\sqrt{\sigma^2h^{2H} - \Sigma_{RS}\Sigma_S^{-1}\Sigma_{RS}^T}}\right)$$
and the non-conditional hit ratio is, for $H\neq1/2$,
\begin{equation}\label{eq:hit}
\rho= 1-\frac{1}{\pi}\arctan\left(\sqrt{\frac{\sigma^2h^{2H}}{\Sigma_{RS}\Sigma_S^{-1}\Sigma_{RS}^T}-1}\right),
\end{equation}
where $\Sigma_{S}$ and $\Sigma_{RS}$ are the same as in Proposition~\ref{pro:NP_incr} and $N$ is the standard Gaussian cumulative distribution function.
\end{thm}

The proof of Theorem~\ref{th:hitNrdt} is postponed in Appendix~\ref{sec:hitNrdt}.

The case $H= 1/2$, for which the fBm is a standard Brownian motion, is of particular interest. In this case, the $1\times 1$ matrix $\Sigma_R$ contains the element $\sigma^2 h$, the  $1\times n$ matrix $\Sigma_{RS}$ only contains zeros, and the $n\times n$ matrix $\Sigma_S$ is diagonal equal to $\text{diag}(\sigma^2(\delta_1-\delta_0),...,\sigma^2(\delta_n-\delta_{n-1}))$. As a consequence, according to Theorem~\ref{th:hitNrdt}, the conditional hit ratio $\rho^c(y)$ is 0.5, whatever $y$. This is consistent with the fact that the standard Brownian motion is a martingale. Besides, any value of $H$ different from $1/2$ trivially leads to $\rho^c>0.5$ since it is the image by $N$ of a positive number. Regarding the non-conditional hit ratio $\rho$, equation~\ref{eq:hit} is not defined for $H=1/2$, but we can easily determine that the limit of $\rho$, when $H\rightarrow 1/2$, is also 0.5. Once, again, any value of $H$ different from $1/2$ leads to $\rho>0.5$, that is to more frequent good predictions than bad predictions, as one can see in Figure~\ref{fig:HitRatio}.

\begin{figure}[htbp]
	\centering
		\includegraphics[width=0.75\textwidth]{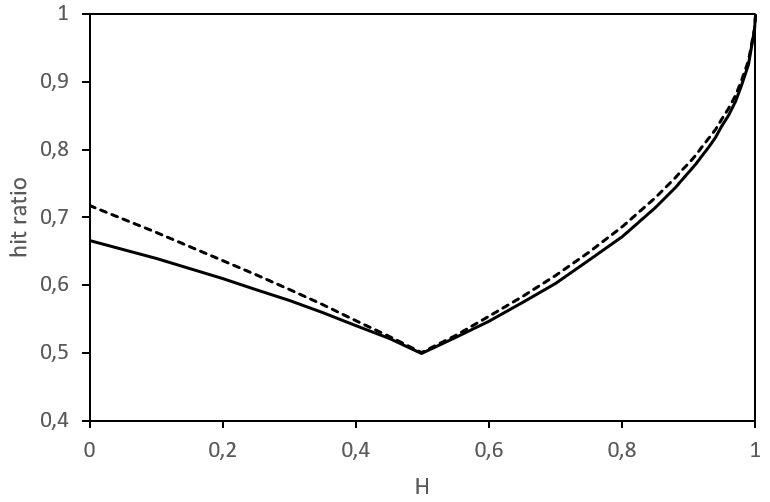} 
\begin{minipage}{0.7\textwidth}\caption{Theoretical non-conditional hit ratio for the fBm with one time lag and $\delta_1=h$ (continuous line) and for four time lags and $(\delta_0,...,\delta_4)=(0,0.5,1,2,3)$ (dotted line).}
	\label{fig:HitRatio}
\end{minipage}
\end{figure}

In the rest of the paper, the theoretical hit ratio on which we focus is $\rho$, the non-conditional hit ratio. Indeed, in a practical application, we are inclined to choose a specific forecast setting, that is a number of lags and a duration for these lags. The maximization of a theoretical hit ratio ex ante must lead the choice of this setting. This choice is not that simple if one wants to maximize the conditional hit ratio, insofar as one must numerically test the best combination of lagged observations to finally retain the one maximizing $\rho^c$. A legitimate question then arises: Why not retaining all the tested lagged observations? If one uses $\rho$ instead of $\rho^c$, the choice of time lags is simpler and will always be the same, provided that the estimated $H$ is also the same, because the choice of time lags in this case does not depend on the particular value of each past observation but only on the estimated $H$.

We see in Figure~\ref{fig:HitRatio} that Hurst exponents closer to 0 or 1 than to $1/2$ lead to higher hit ratios. Hurst exponents $H$ above $1/2$ also lead to more accurate predictions than when the Hurst exponent is $1-H$. In particular, $H$ close to 1, that is a very persistent series of returns, leads to the highest possible hit ratio, whereas $H$ close to 0 only leads to a hit ratio of 0.67 if the forecast is based on one past return ($n=1$). In other words, it is globally more difficult to make good predictions when $H<1/2$ than when $H>1/2$ but predictions are however good in average, with a hit ratio higher than 0.5 in both cases. This contradicts a part of the econophysics literature which tends to consider that predictions are relevant only for $H>1/2$. We can even stress that predictions will be more accurate for $H=0.2$ than for $H=0.6$.

In the particular case where $n=1$, for which we have already provided an explicit expression for $\beta_1$ in equation~\ref{eq:beta1}, the hit ratio also has a simpler expression:
$$\rho=1-\frac{1}{\pi}\arctan\left(\sqrt{\frac{1}{\beta_1^2}\left(\frac{h}{\delta_1}\right)^{2H}-1}\right)$$
because $\Sigma_{RS}\Sigma_S^{-1}\Sigma_{RS}^T$, in this case, is equal to $\beta_1\Sigma_{RS}^T=\beta_1^2\Sigma_S=\sigma^2\beta_1^2\delta_1^{2H}$.

In a systematic investment perspective, a forecast based on only one past return is often not enough to generate a profitable strategy because of too low hit ratios, in particular because the Hurst exponent in finance is almost never in the interval $(0.75,1)$, where we observe the highest theoretical hit ratios. In order to build a systematic strategy, the prediction is to be improved by including several past returns. Even without optimizing the duration of these past returns, one sees in Figure~\ref{fig:HitRatio} that adding arbitrary time lags in the set $\Delta$ indeed increases the hit ratio. This effect is very limited when the Hurst exponent is higher than $1/2$, but it is promising for values of $H$ below $1/2$.

In this investment perspective, the question of the best choice for $n$, the number of past returns to be included in the predictor, is an important topic. Theoretically, the optimal $n$ maximizing the theoretical hit ratio is $+\infty$. However, considering a predictor based on a big number of past returns leads in general to two pitfalls: a higher computation time and overfitting. The first point is a major issue in high-frequency trading. To avoid the second drawback, one prefers minimizing information criteria, such as AIC or BIC, instead of simply maximizing an ex-ante hit ratio. Such criteria are based on the likelihood of the model and penalize the number of parameters. Regarding the likelihood, reality is always more complex than any model, so adding time lags will first improve the likelihood but it may also deteriorate it for a too big number of time lags due to a discrepancy between reality and the fBm specification. Regarding the number of parameters, it does not depend on the number of time lags if the model is an fBm. Indeed, the weights $\beta_i$ applied to past returns all depend on the parameters of the fBm. 

The choice of the $n$ minimizing AIC or BIC thus depends on the dataset. We will see an illustration of this problem in the empirical part of the paper, in Section~\ref{sec:rough}. Once $n$ is selected, the duration of the $n$ past price returns is to be selected too. We will see in the next subsection how this can be achieved using the theoretical hit ratio.

\subsection{Optimal duration of time lags}\label{sec:optim}

The question of the optimal duration of time lags is to be related to the question of the optimal sampling of high-frequency observations of asset prices, with which the econometric literature has already dealt without any fBm assumption~\cite{BR,Oomen}. Observations are equally spaced in time in this econometric approach. This will not be the case in our fBm-based predictor.

We consider a fixed number $n\geq 1$ of past returns to be used in the predictor defined in Proposition~\ref{pro:NP_incr}. We write them $R_{t-\delta_1,t-\delta_0},...,R_{t-\delta_n,t-\delta_{n-1}}$, with $\forall i<j$, $\delta_i<\delta_{j}$, $\delta_i\geq 0$, and $\delta_0=0$. Given a fixed prediction horizon $h>0$, the vector of time lags $[\delta_0\ ...\ \delta_n]^T$ maximizing the hit ratio $\rho$ of Theorem~\ref{th:hitNrdt} is noted $[\delta_0^{\star}\ ...\ \delta_n^{\star}]^T$. A numerical optimization of $\rho$ provides these optimal time lags. This numerical study leads to the following observation:
$$\forall i\in\llbracket 1,n\rrbracket,\ \delta_i^{\star}=\frac{h^2}{\delta_{n+1-i}^{\star}}.$$

We briefly focus on the case $n=1$. Figure~\ref{fig:HitDelta} displays the value of the hit ratio as a function of $\delta_1/h$. The optimum in $\delta_1/h=1$ is clear and it also appears that the hit ratio is the same for $\delta_1/h$ and $h/\delta_1$. Therefore, the optimal time lag for the past return used in the forecast is equal to the forecast horizon: $\delta_1=h$. In other words, to forecast the future daily return, one has to consider the last observed daily return. For a monthly return to be forecast, the last observed monthly return is the optimal input of the fBm-based prediction model. This is consistent with the existing literature~\cite{NP}. 

\begin{figure}[htbp]
	\centering
		\includegraphics[width=0.6\textwidth]{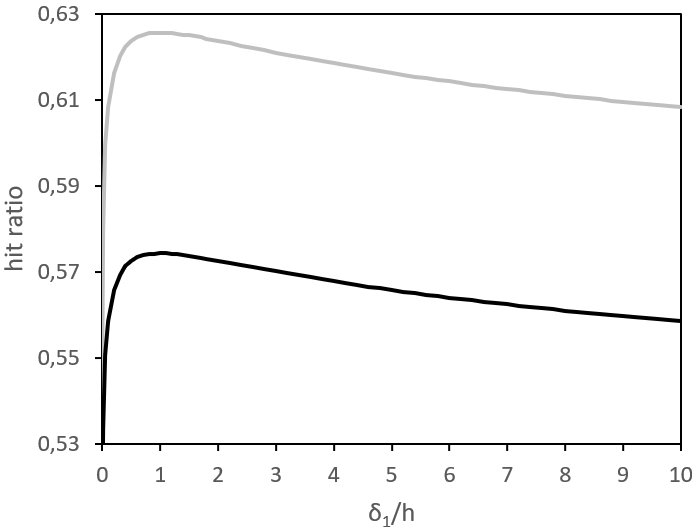} 
\begin{minipage}{0.7\textwidth}\caption{Hit ratio with $H=0.65$ (black) and $H=0.15$ (grey) for the fBm, for various values of the ratio $\delta_1/h$ and $n=1$.}
	\label{fig:HitDelta}
\end{minipage}
\end{figure}

Other cases than $n=1$ provide new insights. Increasing $n$ of one unit does indeed not mean adding one time lag to the set of time lags optimized in the case $n-1$. For example, if $n=2$, $h$ is not one of the optimal time lags. Instead, if we consider $h=1$ for simplifying and $H=0.65$, the optimal lags are $(\delta_1,\delta_2)=(0.29,3.45)$. For a monthly forecast, so for a forecast horizon of 22 days, the optimal time lags to be considered are thus roughly 6 and 76 days. For a daily forecast, intraday data are to be taken into account. If $n$ is odd (respectively even), one has to consider $(n-1)/2$ (resp. $n/2$) time lags lower than $h$, the same quantity bigger than $h$, and one (resp. zero) time lag equal to $h$. We gather in Tables~\ref{tab:lag65} and~\ref{tab:lag15} the optimal lags for two values of $H$ and various $n$. We provide them for $h=1$. Multiplying them by $h$ gives the optimal time lags for a forecast horizon equal to $h$.

\begin{table}[htbp]
\centering
\begin{tabular}{c|cccccc|c}
\hline
$n$ & $\delta_1^{\star}$ & $\delta_2^{\star}$ & $\delta_3^{\star}$ & $\delta_4^{\star}$ & $\delta_5^{\star}$ & $\delta_6^{\star}$ & $\rho$ \\
\hline
 1 & 1.000 & & & & & & $57.42\%$ \\
 2 & 0.289 & 3.454 & & & & & $58.14\%$ \\
 3 & 0.127 & 1.000 & 7.896 & & & & $58.72\%$ \\
 4 & 0.067 & 0.458 & 2.185 & 14.979 & & & $58.90\%$ \\
 5 & 0.039 & 0.253 & 1.000 & 3.949 & 25.407 & & $58.99\%$ \\
 6 & 0.025 & 0.156 & 0.562 & 1.780 & 6.411 & 39.919 & $59.05\%$ \\
\hline
\end{tabular}
\begin{minipage}{0.7\textwidth}\caption{For various numbers $n$ of lagged returns, optimal time lags and corresponding theoretical hit ratio. The dynamic is an fBm of Hurst exponent 0.65 and the forecast horizon is $h=1$.}
\label{tab:lag65}
\end{minipage}
\end{table}

\begin{table}[htbp]
\centering
\begin{tabular}{c|cccccc|c}
\hline
$n$ & $\delta_1^{\star}$ & $\delta_2^{\star}$ & $\delta_3^{\star}$ & $\delta_4^{\star}$ & $\delta_5^{\star}$ & $\delta_6^{\star}$ & $\rho$ \\
\hline
 1 & 1.000 & & & & & & $62.56\%$ \\
 2 & 0.367 & 2.726 & & & & & $64.34\%$ \\
 3 & 0.193 & 1.000 & 5.168 & & & & $65.72\%$ \\
 4 & 0.120 & 0.539 & 1.856 & 8.365 & & & $66.29\%$ \\
 5 & 0.081 & 0.341 & 1.000 & 2.933 & 12.347 & & $66.66\%$ \\
 6 & 0.058 & 0.236 & 0.637 & 1.570 & 4.241 & 17.170 & $66.91\%$ \\
\hline
\end{tabular}
\begin{minipage}{0.7\textwidth}\caption{For various numbers $n$ of lagged returns, optimal time lags and corresponding theoretical hit ratio. The dynamic is an fBm of Hurst exponent 0.15 and the forecast horizon is $h=1$.}
\label{tab:lag15}
\end{minipage}
\end{table}

Another question is important for defining a systematic trading strategy: what is the best between defining each price return on the one hand between $t-\delta_i$ and $t-\delta_{i-1}$ and, on the other hand, between $t-\delta_i$ and $t$? It is worth noting that if we compare the two vectors $[R_{t-\delta_1,t-\delta_0}\ ...\ R_{t-\delta_n,t-\delta_{n-1}}]^T$ and $[R_{t-\delta_1,t}\ ...\ R_{t-\delta_n,t}]^T$, the weight vector $[\beta_1\ ...\ \beta_n]^T$ associated to each vector of price returns is different, but the accuracy of the corresponding predictor, defined as the non-conditional hit ratio, is the same. Consequently, the optimal lags are also the same between the two versions. Whatever the way one splits the time lags among the different returns, provided that there is no redundancy, the non-conditional quality of the forecast will be the same. All the information used by the predictor comes in fact from the list of time lags. One must thus concentrate one's efforts on the selection of $\Delta$ instead of on its specific division as bounds of all the price returns.

\section{Statistical arbitrage: to predict or not to predict}\label{sec:statarb}

The theoretical results in Section~\ref{sec:forecast}, about hit ratios for predictors of an fBm, are encouraging. Indeed, many values of the Hurst exponent lead to an acceptable hit ratio in the perspective of building a trading strategy. However a more precise look at the predicted increments show some forecasts close to zero which may affect the performance of a trading strategy for two reasons. First, the smaller the predictions in absolute value, the higher the incertitude about the sign of the future return. The second reason is almost a tautology: the smaller the future price returns of an asset in absolute value, the lower the financial performance a trader can expect by following an investment strategy using this asset. Because of these two reasons, a trader may want to discard predictions which are close to zero. We are thus interested in understanding how the quality of the forecast evolves when one applies a thresholding to the predictor introduced in equation~\eqref{eq:r_hat}. Next, we want not to limit our analysis to hit ratios but to introduce more financial criteria, such as expected performance and risk of a simple trading strategy based on the thresholding concept. This strategy consists in buying the asset if the predicted return is significantly above zero, short selling the asset if it is significantly below, and keeping a neutral position if the prediction is too close to zero for the trader to be confident in the predicted sign of the future return.

\subsection{Threshold and hit ratio}\label{sec:ThreshHit}

We still suppose that log-prices follow an fBm. Thanks to this assumption, we can build predictors following equation~\eqref{eq:r_hat}, in which we can input a number $n$ of lagged observed returns. We want to determine the theoretical forecasting performance of such a predictor when one discards the predicted values which are close to zero. In other words, one defines a threshold $\theta\geq 0$ and ignores the predictions lower than $\theta$ in absolute value. We thus introduce a thresholding function 
$$f_{\theta}:x\in\mathbb R\mapsto x\indic_{\{|x|\geq \theta\}},$$
so that the considered predictor is now $f_{\theta}\left(\hat R_{t,t+h|\Delta}\right)$.

The hit ratio we used in Section~\ref{sec:forecast} is not appropriate in this new framework. Indeed, the hit ratio only supposes two states: on the one hand, the good prediction of the sign of the future return and, on the other hand, the bad prediction. We now use a ternary classification of predictions: good, bad, and zero prediction. The statistics depicting the forecast accuracy of the new predictor can be based on the probability of each of these three states.

\begin{defn}\label{def:probaTernary}
Let $\hat R_{t,t+h|\Delta}$ be a predictor of $R_{t,t+h}$ based on the vector $S$ of past returns, as defined by equation~\eqref{eq:r_hat}. Given $\theta\geq 0$, we transform $\hat R_{t,t+h|\Delta}$ in a new predictor $f_{\theta}\left(\hat R_{t,t+h|\Delta}\right)$. The non-conditional probability of good sign forecast is
$$p^+(\theta)=\proba\left[f_{\theta}\left(\hat R_{t,t+h|\Delta}\right)R_{t,t+h}>0\right],$$ the non-conditional probability of bad sign forecast is 
$$p^-(\theta)=\proba\left[f_{\theta}\left(\hat R_{t,t+h|\Delta}\right)R_{t,t+h}<0\right],$$
and the non-conditional probability of zero forecast is
$$p^0(\theta)=\proba\left[f_{\theta}\left(\hat R_{t,t+h|\Delta}\right)=0\right].$$
\end{defn}

We stress the fact that we are only interested here in non-conditional probabilities. Taking into account conditional probabilities may lead to better predictors, in which more appropriate weights are used for lagged returns, and thus to a better performance. Nevertheless, the conditional approach is very specific to the data observed and is a useless refinement in our perspective. The purpose of this paper is indeed not the fine tuning of a trading strategy but, instead, it is about setting general results concerning the relevance of fBm-based predictions in finance. 

The formula for the probabilities introduced in Definition~\ref{def:probaTernary} are expressed in Theorem~\ref{th:ternary}, in which the fBm assumption for log-prices is used.

\begin{thm}\label{th:ternary}
Let $X_t$ be an fBm of Hurst exponent $H\in(0,1/2)\cup(1/2,1)$ and $\theta\geq 0$. Let $h>0$ and $\hat R_{t,t+h|\Delta}$ be a predictor of $R_{t,t+h}$ based on the vector $S$ of past returns, as defined by equation~\eqref{eq:r_hat}. The non-conditional probabilities of sign forecast, as introduced in Definition~\ref{def:probaTernary}, are provided by a Taylor expansion, when $\theta$ is in the neighbourhood of 0, for $p^+(\theta)$:
$$p^+(\theta)=1-N\left(\frac{\theta}{a}\right)+\frac{1}{\pi}\arctan\left(\frac{a}{b}\right)-\frac{\theta^2}{2\pi ab}+\left(\frac{1}{ab^3}+\frac{3}{a^4b}\right)\frac{\theta^4}{24\pi}+\mathcal O (\theta^6)$$
and for $p^-(\theta)$:
$$p^-(\theta)=1-N\left(\frac{\theta}{a}\right)-\frac{1}{\pi}\arctan\left(\frac{a}{b}\right)+\frac{\theta^2}{2\pi ab}-\left(\frac{1}{ab^3}+\frac{3}{a^4b}\right)\frac{\theta^4}{24\pi}+\mathcal O (\theta^6),$$
as well as by the following exact formula for $p^0(\theta)$:
$$p^0(\theta)=-1+2N\left(\frac{\theta}{a}\right),$$
where $a= \sqrt{\Sigma_{RS}\Sigma_S^{-1}\Sigma_{RS}^T}$, $b=\sqrt{\sigma^2h^{2H} - \Sigma_{RS}\Sigma_S^{-1}\Sigma_{RS}^T}$, and where $\Sigma_{S}$ and $\Sigma_{RS}$ are the same as in Proposition~\ref{pro:NP_incr}.
\end{thm}

The proof of Theorem~\ref{th:ternary} is postponed in Appendix~\ref{sec:ternary}.

Theorem~\ref{th:ternary} provides only Taylor expansions of $p^+(\theta)$ and $p^-(\theta)$, whereas we have an exact formula for $p^0(\theta)$. However, beyond the expansions, we can describe how $p^+(\theta)$ and $p^-(\theta)$ evolve with $\theta$, whatever $H$. Indeed, both $p^+$ and $p^-$ are trivially decreasing functions. Moreover, $p^+$ decreases more rapidly than $p^-$, as stated in Proposition~\ref{pro:derivP}. 

\begin{pro}\label{pro:derivP}
With the assumptions of Theorem~\ref{th:ternary}, for $\theta\geq 0$ we have:
$$\frac{dp^+}{d\theta}(\theta) \leq \frac{dp^-}{d\theta}(\theta) \leq 0.$$
\end{pro}

\begin{proof}
We remark from equations~\eqref{eq:p+} and~\eqref{eq:p-}, in the proof of Theorem~\ref{th:ternary}, that the derivative of $p^+$ and $p^-$ is:
$$\left\{\begin{array}{lll}
\frac{d}{d\theta}p^+(\theta) & = & -\frac{2}{a}N\left(\frac{\theta}{b}\right)g\left(\frac{\theta}{a}\right) \\
\frac{d}{d\theta}p^-(\theta) & = & -\frac{2}{a}\left(1-N\left(\frac{\theta}{b}\right)\right)g\left(\frac{\theta}{a}\right).
\end{array}\right.$$
The two derivatives are negative and the ratio of the second one over the first one is $N(\theta/b)^{-1}-1$, which is in the interval $[0,1]$ for $\theta\geq 0$, meaning that the derivative of $p^+$ with respect to $\theta\geq 0$ is greater in absolute value than the one of $p^-$.
\end{proof}

In Theorem~\ref{th:ternary}, when the threshold is $\theta=0$, the zero forecast probability is $p^0(0)=0$. This particular case corresponds to the binary classification of predictions of Section~\ref{sec:forecast}. The non-conditional hit ratio $\rho$ introduced in Definition~\ref{def:hit} corresponds to $p^+(0)$, whose formula in Theorem~\ref{th:ternary} matches the formula of $\rho$ in Theorem~\ref{th:hitNrdt}.

The greater the threshold $\theta$, the greater the value of $p^0(\theta)$, whose limit, when $\theta\rightarrow+\infty$, is 1. As a consequence, if $\theta\neq 0$, we do not have a binary classification between good and bad forecast but a ternary one, the third state being the zero prediction because of the thresholding. The hit ratio is thus inappropriate. One can think of several alternatives exploiting $p^+(\theta)$, $p^-(\theta)$, and $p^0(\theta)$.

The first natural alternative to the hit ratio consists in considering the proportion of good predictions with respect to the non-zero predictions, that is $p^+(\theta)/(p^+(\theta)+p^-(\theta))$. This ratio converges asymptotically towards its maximal value, 1. Indeed, from equations~\eqref{eq:p+} and~\eqref{eq:p-} in the proof of Theorem~\ref{th:ternary}, we have $p^+(\theta)+p^-(\theta)=2N(-\theta/a)$ and
$$\begin{array}{ccl}
\frac{p^+(\theta)}{p^+(\theta)+p^-(\theta)} & = & \frac{1}{2} + \frac{1}{2N(-\theta/a)} \int_{\theta/a}^{+\infty}{\left(N\left(\frac{a}{b}u\right)-N\left(-\frac{a}{b}u\right)\right)g(u)du} \\
 & \overset{\theta\rightarrow +\infty}{\sim} & \frac{1}{2} + \frac{1}{2N(-\theta/a)} \int_{\theta/a}^{+\infty}{g(u)du} \\
 & = & \frac{1}{2} + \frac{1}{2N(-\theta/a)} \left(1-N(\theta/a)\right) \\
 & = & 1.
\end{array}$$ 
In a financial perspective, this accuracy metric is not suitable, because maximizing it leads to selecting $\theta=+\infty$ and thus to discarding all the predictions.

In order to penalize high values of the zero forecast probability, on can consider the alternative ratio $p^+(\theta)/(p^+(\theta)+p^-(\theta)+p^0(\theta))=p^+(\theta)$, but it is monotonically decreasing in $\theta$, as stated by Proposition~\ref{pro:derivP}, so that the optimal threshold would be $\theta=0$, whatever $H$. This solution thus does not address the problem introduced in this section about discarding the less certain predictions.

Alternatively, $p^+(\theta)-p^-(\theta)$ is a signed forecast accuracy metric. This metric corresponds to the average number of coins earned in a game in which one earns one coin for a good forecast, loses one for a bad one, and there is no coin transaction if one restrains from playing this forecasting game. According to Proposition~\ref{pro:derivP}, this ratio is monotonically decreasing in $\theta$, like the previous ratio.

All the natural extensions of the hit ratio we can think of are inappropriate insofar as they lead to an optimal threshold $\theta$ either equal to 0 or to infinity. We propose in the next subsection an alternative accuracy metric more consistent with a financial application. Indeed, all the extended hit ratios proposed above only rely on probabilities of making a good, a bad, or a small forecast. In particular, they do not consider the amplitude of the predicted price returns. In an investment perspective, forecasting properly large price returns is though more fruitful than forecasting properly small ones. 

\subsection{Optimal threshold and risk-adjusted performance}\label{sec:seuilOptim}

We now consider the financial performance corresponding to a prediction. A trader is more eager to forecast accurately large price variations than small ones. Therefore, the trader's objective is not minimizing an MSE or maximizing a hit ratio, but maximizing instead a risk-adjusted performance of an investment based on these predictions. A proper evaluation of the predictor must indeed weight the 0-1 contribution of each prediction, which appears in the hit ratio, with the amplitude of the prediction. 

The trading strategy we study here is the ternary strategy exposed above: buying a fixed quantity of the asset if one predicts a positive price return, short selling the same quantity if one predicts a negative return, not investing if the expected return of the asset is below the threshold $\theta$ in absolute value. The trading frequency of the strategy is equal to $h>0$, the forecast horizon. Using the same notations as above and given an amount of money hold by the trader and totally invested in this strategy, the return at horizon $h$ will thus be:
$$R^{\text{strat}}_{t,t+h}(\theta)=\left\{\begin{array}{cl}
0 & \text{if } f_{\theta}\left(\hat R_{t,t+h|\Delta}|\right)=0 \\
R_{t,t+h} & \text{if } f_{\theta}\left(\hat R_{t,t+h|\Delta}|\right)>0 \\
-R_{t,t+h} & \text{if } f_{\theta}\left(\hat R_{t,t+h|\Delta}|\right)<0. 
\end{array}\right.$$
This return represents the performance to be maximized. At time $t$, one only knows an expectation of this return: 
\begin{equation}\label{eq:perfTernary}
\widetilde R_{t,t+h}(\theta)=\E[R^{\text{strat}}_{t,t+h}(\theta)].
\end{equation}

Regarding the risk statistic of this investment, we consider a lower absolute semi-deviation risk measure. It is defined as the average absolute deviation below zero of the return of the strategy, in other words it is the average loss:
\begin{equation}\label{eq:riskTernary}
\widetilde \sigma_{t,t+h}^-(\theta)=-\E\left[\min(0,R^{\text{strat}}_{t,t+h}(\theta))\right].
\end{equation}
Such a risk measure is less widespread in the asset management industry than the volatility, but it is more relevant since it does not incorporate positive deviations which, by definition, are not risky. This risk measure appears for example at the denominator of the kappa ratio of order one, which is a general risk-performance ratio including the popular  Sharpe ratio as a particular case~\cite{VanHarlow}.

Theorem~\ref{th:ternaryFin} provides a formula for $\widetilde R_{t,t+h}(\theta)$ and for $\widetilde \sigma_{t,t+h}^-(\theta)$ when the log-prices follow an fBm.

\begin{thm}\label{th:ternaryFin}
Let $X_t$ be an fBm of parameters $H\in(0,1/2)\cup(1/2,1)$ and $\sigma>0$. Let $h>0$ and $\hat R_{t,t+h|\Delta}$ be a predictor of $R_{t,t+h}$ based on the vector $S$ of past returns, as defined by equation~\eqref{eq:r_hat}. Given $\theta\geq 0$, the average return of the ternary strategy, $\widetilde R_{t,t+h}(\theta)$, and the corresponding lower absolute semi-deviation risk measure, $\widetilde \sigma^-_{t,t+h}(\theta)$, defined in equations~\eqref{eq:perfTernary} and~\eqref{eq:riskTernary}, are equal to:
$$\widetilde R_{t,t+h}(\theta)= 2ag\left(\frac{\theta}{a}\right)$$
and
$$\widetilde \sigma^-_{t,t+h}(\theta)= - 2aN\left(-\frac{\theta}{b}\right)g\left(\frac{\theta}{a}\right) + \sqrt{\frac{2}{\pi}}\sigma h^{H}N\left(-\theta\sqrt{\frac{1}{a^2}+\frac{1}{b^2}}\right),$$
where $a= \sqrt{\Sigma_{RS}\Sigma_S^{-1}\Sigma_{RS}^T}$, $b=\sqrt{\sigma^2h^{2H} - \Sigma_{RS}\Sigma_S^{-1}\Sigma_{RS}^T}$, and $g$ is the Gaussian probability density function.
\end{thm}

The proof of Theorem~\ref{th:ternaryFin} is postponed in Appendix~\ref{sec:ternaryFin}.

In the case where $\theta=0$, that is when the trader does not discard any forecast, the performance and the risk of the strategy are simplified:
$$\left\{\begin{array}{ccl}
\widetilde R_{t,t+h}(0) & = & a\sqrt{\frac{2}{\pi}} \\
\widetilde \sigma^-_{t,t+h}(0) & = & \frac{\sigma h^H-a}{\sqrt{2\pi}},
\end{array}\right.$$
in which we recognize, for $\widetilde R_{t,t+h}(0)$, the expected value of $a|U|$, where $U\sim\mathcal N(0,1)$. Both $\widetilde R_{t,t+h}(\theta)$ and $\widetilde \sigma^-_{t,t+h}(\theta)$ are monotonically decreasing in $\theta$: discarding predictions tends to diminish both the expected gain and the risk. At the limit, when $\theta\rightarrow +\infty$, both $\widetilde R_{t,t+h}(\theta)$ and $\widetilde \sigma^-_{t,t+h}(\theta)$ tend to zero: the threshold is so high that no investment decision is made by the trader.

In order to find a good balance between performance and risk, we define an expected risk-adjusted performance, which simply consists in the expected return of the strategy, penalized by the corresponding risk:
\begin{equation}\label{eq:riskAdj}
\widetilde R_{\lambda}(\theta)=\widetilde R_{t,t+h}(\theta)-\lambda\widetilde \sigma_{t,t+h}^-(\theta).
\end{equation}
In equation~\eqref{eq:riskAdj}, we omit the subscript in $t$ and $t+h$ for the risk-adjusted performance $\widetilde R_{\lambda}(\theta)$, in order not to overload the notations. But one has to keep in mind that $\widetilde R_{\lambda}(\theta)$ depends on $h$. The parameter $\lambda$ in equation~\eqref{eq:riskAdj} plays the role of a risk aversion when its value is positive. A negative value of $\lambda$ would indicate a risk-seeking behaviour. In particular, if $\lambda=-1$, $\widetilde R_{\lambda}(\theta)$ is simply the upper absolute semi-deviation, $\E[\max(0,R^{\text{strat}}_{t,t+h}(\theta))]$. In what follows, we focus on the case $\lambda\geq 0$. Because of the particular risk measure we have chosen, the risk-adjusted performance corresponds to the expected performance of the strategy in a distorted probability in which bad outcomes, that is to say negative returns, are overweighted. Indeed,
$$\widetilde R_{\lambda}(\theta)=\int_0^1{Q_{R^{\text{strat}}_{t,t+h}(\theta)}(p)d\mu(p)},$$
where $Q_{Z}(p)$ is the quantile of probability $p$ of a variable $Z$ and $\mu$ is the probability distortion function defined by
$$\mu:x\in[0,1]\longmapsto \left\{\begin{array}{ll}
\frac{(1+\lambda)x}{1+\lambda p^-(\theta)} & \text{if }  x<p^-(\theta) \\
\frac{x + \lambda p^-(\theta)}{1+\lambda p^-(\theta)} & \text{else.}
\end{array}\right.$$
In particular, if $\lambda=1$, the probability distortion doubles the probability of negative returns, relatively to the probability of positive returns.

Figure~\ref{fig:Ternary_PenalizedR} displays the risk-adjusted performance for various values of risk aversion $\lambda$ and threshold $\theta$. Since both the expected return and the risk measure tend toward zero when $\theta\rightarrow+\infty$, $\widetilde R_{\lambda}(\theta)$ also tends toward zero for high threshold values. When $\lambda=0$, the maximum risk-adjusted performance is reached for $\theta=0$. For $\lambda>0$, we observe a global maximum for a value $\theta>0$. This threshold is important insofar as it constitutes an optimal threshold for a trader whose aim would be the maximization of the risk-adjusted performance:
$$\theta^{\star}_{\lambda}=\underset{\theta\geq 0}{\text{argmax}}\ \widetilde R_{\lambda}(\theta).$$

\begin{figure}[htbp]
	\centering
		\includegraphics[width=0.6\textwidth]{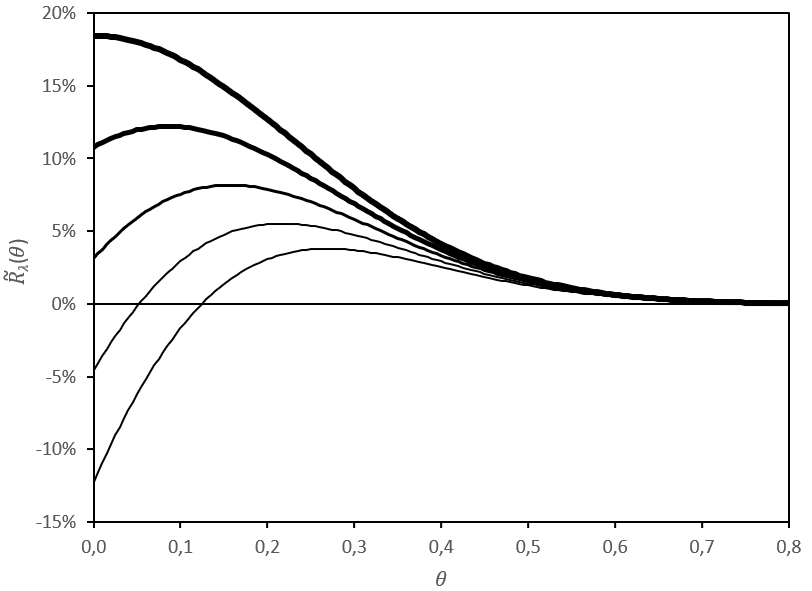} 
\begin{minipage}{0.7\textwidth}\caption{Theoretical penalized expected return $\widetilde R_{\lambda}(\theta)$ as a function of the threshold $\theta$, for various values of risk aversion $\lambda$, namely, from fatter to thinner, 0, 0.25, 0.5, 0.75, 1. The fBm has the parameters $H=0.65$ and $\sigma=1$. The forecast horizon is $h=1$ and the number of lagged returns in the predictor is $n=1$.}
	\label{fig:Ternary_PenalizedR}
\end{minipage}
\end{figure}

Figure~\ref{fig:Ternary_Lambda} illustrates the link between $\lambda$ and $\theta^{\star}_{\lambda}$. The higher the risk aversion $\lambda$, the higher the optimal threshold $\theta^{\star}_{\lambda}$. Moreover, a higher threshold means that the decision not to invest is more frequent. Figure~\ref{fig:Ternary_Lambda} thus shows $p^0(\theta^{\star}_{\lambda})$ as an increasing function of $\lambda$.

\begin{figure}[htbp]
	\centering
		\includegraphics[width=0.49\textwidth]{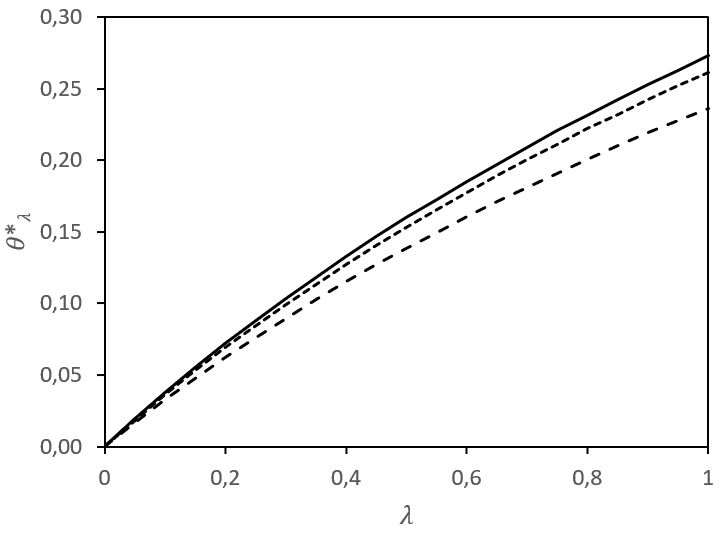} 
		\includegraphics[width=0.49\textwidth]{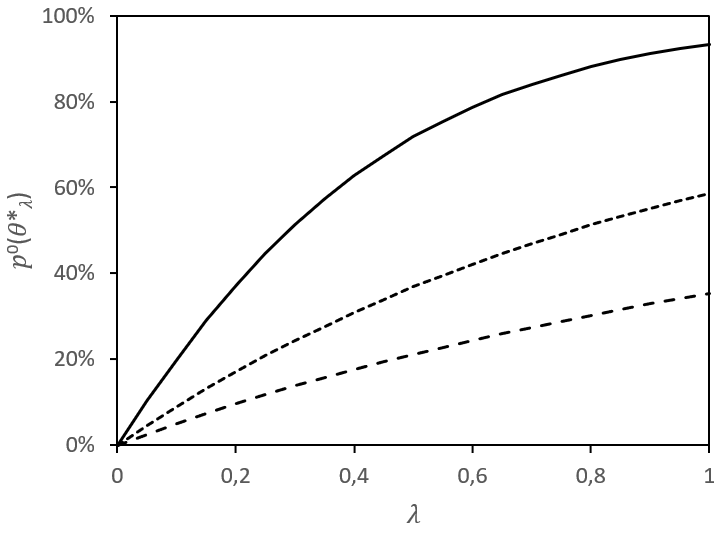} 
\begin{minipage}{0.7\textwidth}\caption{Theoretical optimal threshold $\theta^{\star}_{\lambda}$ (left) and corresponding non-conditional zero forecast probability $p^{0}(\theta^{\star}_{\lambda})$ (right) as functions of the risk aversion $\lambda$. The fBm has the parameters $\sigma=1$ and, for $H$, respectively from the top curve to the bottom curve, 0.6, 0.7, and 0.8. The forecast horizon is $h=1$ and the number of lagged returns in the predictor is $n=1$.}
	\label{fig:Ternary_Lambda}
\end{minipage}
\end{figure}

We also observe that the optimal threshold strongly depends on the value of the Hurst exponent, as displayed in Figure~\ref{fig:Ternary_H}. In particular, it is asymmetric with respect to $H=1/2$. Whatever the risk aversion, the highest optimal thresholds are reached for Hurst exponents close to $1/2$, but the optimal threshold decreases toward 0 when $H$ is close to 1, whereas a threshold significantly different from 0 is indicated when $H$ is close to 0. This shows a relatively higher difficulty to make financially performing predictions when $H<1/2$ than when $H>1/2$. Figure~\ref{fig:Ternary_H} also shows $p^0(\theta^{\star}_{\lambda})$ as a function of the Hurst exponent. This zero forecast probability reaches huge values when $H$ is close to $1/2$. For instance, for the risk aversion considered, $\lambda=0.1$, $p^0(\theta^{\star}_{\lambda})$ is twice as large for $H=0.55$ ($40\%$) than for $H=0.6$ ($20\%$).

\begin{figure}[htbp]
	\centering
		\includegraphics[width=0.49\textwidth]{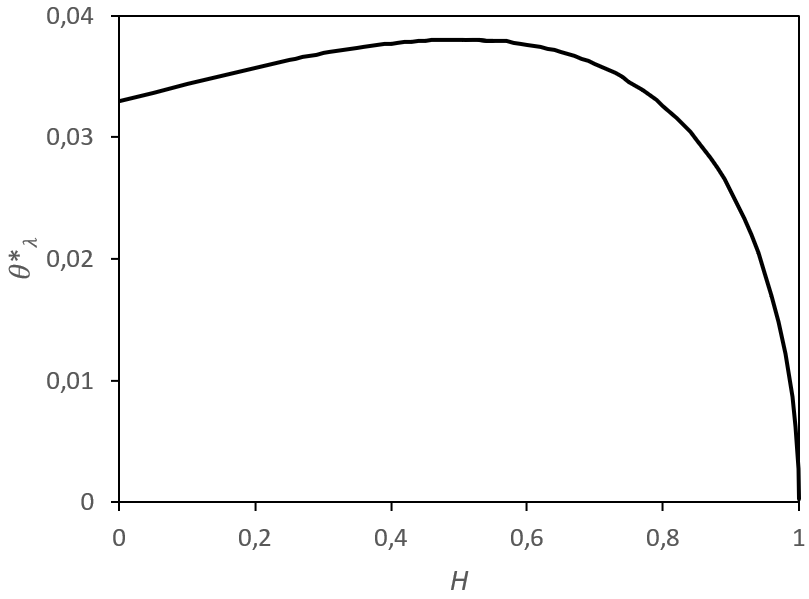} 
		\includegraphics[width=0.49\textwidth]{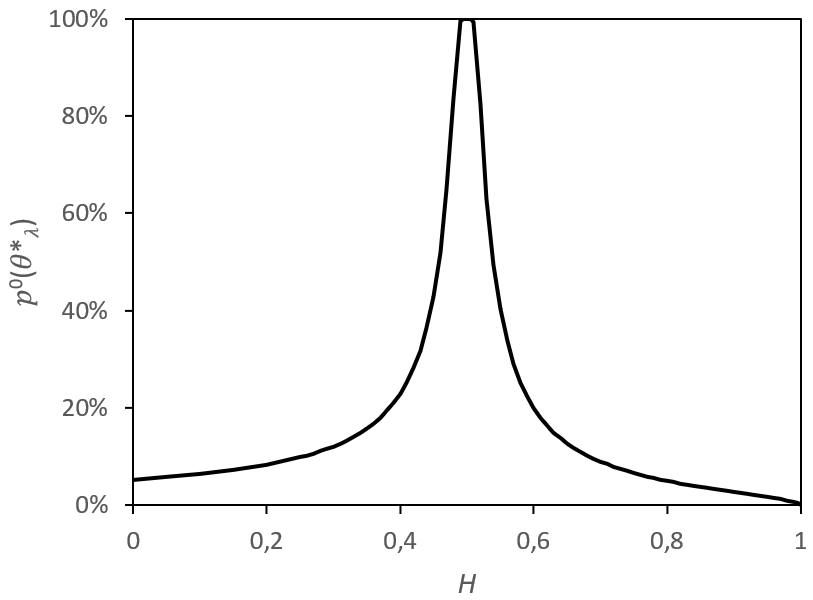} 
\begin{minipage}{0.7\textwidth}\caption{Theoretical optimal threshold $\theta^{\star}_{\lambda}$ (left) and corresponding non-conditional zero forecast probability $p^{0}(\theta^{\star}_{\lambda})$ (right) as functions of the Hurst exponent $H$. The fBm has a volatility parameter $\sigma=1$. The risk aversion is $\lambda=0.1$, the forecast horizon is $h=1$, and the number of lagged returns in the predictor is $n=1$.}
	\label{fig:Ternary_H}
\end{minipage}
\end{figure}

Guasoni and co-authors provide an fBm-based risk-adjusted performance which differs from our approach for several reasons: their risk is a variance, whereas ours is a lower absolute semi-deviation, and they consider trading in continuous time, whereas we more realistically work in discrete time~\cite{GNR,GMR}. They notice an asymmetry of their risk-adjusted performance with respect to $H=1/2$. While the shape of our metric is obviously different from theirs, we also obtain an asymmetry of our risk-adjusted performance, as illustrated by Figure~\ref{fig:RiskAdj}: the worst risk-adjusted performance is for $H=1/2$ and the improvement of this metric is stronger for $H>1/2$ than for $H<1/2$. In the case where $n=1$ and with an optimal lag $\delta_1=h$, we obtain, thanks to equation~\eqref{eq:beta1}, 
$$\begin{array}{ccl}
a & = & \sqrt{\Sigma_{RS}\Sigma_{S}^{-1}\Sigma_{RS}^T} \\
 & = & \sqrt{\beta_1^2\sigma^2\delta_1^{2H}} \\
 & = & \left|\beta_1\sigma\delta_1^H\right| \\
 & = & \sigma h^H\left|2^{2H-1}-1\right|,
\end{array}$$
and finally the risk-adjusted performance without thresholding small predictions is equal to
$$\begin{array}{ccl}
\widetilde R_{\lambda}(0) & = & a\left(\sqrt{\frac{2}{\pi}}+\frac{\lambda}{\sqrt{2\pi}}\right) - \frac{\lambda\sigma h^H}{\sqrt{2\pi}} \\
 & = & \frac{\sigma h^H}{\sqrt{2\pi}}\left[\left|2^{2H-1}-1\right|\left(2+\lambda\right) - \lambda\right].
\end{array}$$
We see in Figure~\ref{fig:RiskAdj} that introducing the optimal thresholding tends to set a floor at 0 for the risk-adjusted performance. As a consequence, when $H$ is close to $1/2$, $\widetilde R_{\lambda}(\theta^{\star}_{\lambda})$ is close to the performance of the strategy without risk adjustment or thresholding, $\widetilde R_{0}(0)$. In this case, the thresholding thus tends to erase the risk. However, when $H$ is close to 0 or 1, $\widetilde R_{\lambda}(\theta^{\star}_{\lambda})$ is close to $\widetilde R_{\lambda}(0)$ and the effect of the thresholding is not significant. 

\begin{figure}[htbp]
	\centering
		\includegraphics[width=0.49\textwidth]{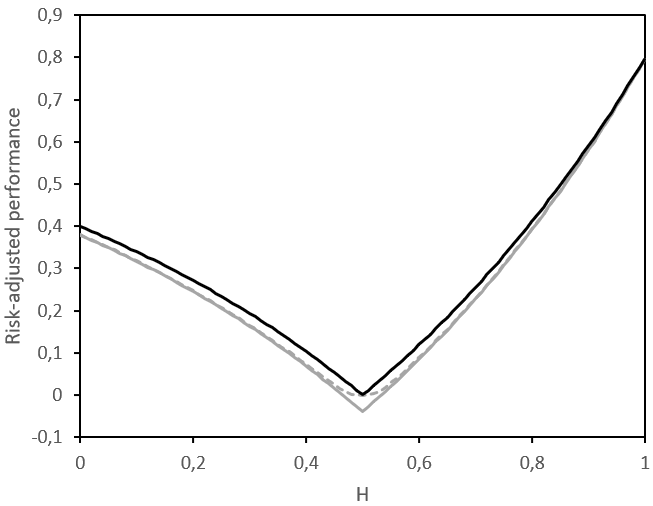} 
		\includegraphics[width=0.49\textwidth]{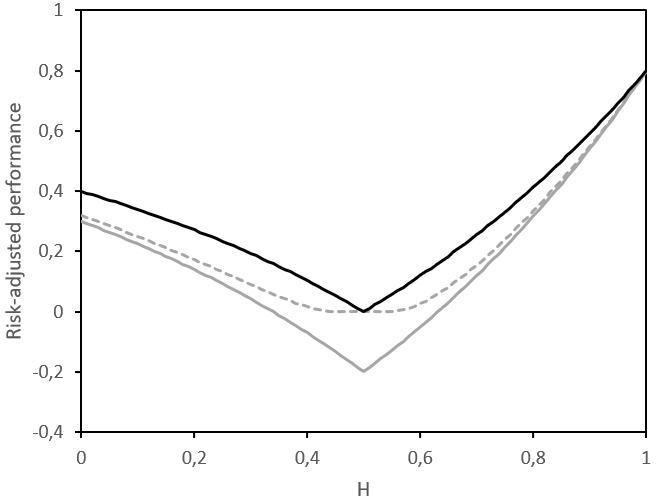} 
\begin{minipage}{0.7\textwidth}\caption{Theoretical penalized expected return $\widetilde R_{0}(0)$ (black line), $\widetilde R_{\lambda}(0)$ (continuous grey line), and $\widetilde R_{\lambda}(\theta^{\star}_{\lambda})$ (dotted grey line), as a function of the Hurst exponent, for $\lambda=0.1$ (left) and $\lambda=0.5$ (right). The fBm has a volatility parameter $\sigma=1$, the forecast horizon is $h=1$ and the number of lagged returns in the predictor is $n=1$.}
	\label{fig:RiskAdj}
\end{minipage}
\end{figure}

\section{Empirical application}\label{sec:empiric}

We now apply the theoretical results exposed above to two kinds of samplings of financial time series: daily sampling of various financial dynamics and high-frequency prices of foreign exchange rates. With daily sampled data, we simply check whether the fBm is a relevant model and we compare the forecasting performance of this model with an autoregressive (AR) model. However, the daily sampling does not make it possible to optimize the time lags in the forecasting procedure as detailed in Section~\ref{sec:optim}. For this purpose, we then use intraday prices and analyse the extent to which the systematic strategy exposed in Section~\ref{sec:statarb} is performing.

\subsection{Daily sampling: fBm vs autoregressive model}\label{sec:rough}

We want to see that the fBm is appropriate for various financial series with a daily time step. We first briefly focus on series of prices of some assets and then we present a more promising extension to series of realized volatilities.

\subsubsection{Daily series of prices}

Regarding the price series, we consider two stock indices, the CAC 40 index and the S\&P 500 index (SPX), as well as one FX rate, the GBPUSD. The series of daily prices starts in January 2000 for the two stock indices and in December 2003 for the GBPUSD. They all end the 12th April 2021.

We want to determine if the fBm-based predictor introduced in Proposition~\ref{pro:NP_incr} is accurate for the time series of the logarithm of the prices, compared to a predictor based on an AR model. We want to make a one-day forecast using daily increments of log-prices, namely $n$ lagged observations, for $n$ varying between one up to six days. In particular, we do not try to optimize the time lags, contrary to the developments of Section~\ref{sec:optim}, in order to make a fair comparison between fBm and AR models. For each value of $n$, the fBm-based forecast is the weighted sum of past increments, as in equation~\eqref{eq:r_hat}. The weights depend both on $n$ and on $H$, the Hurst exponent. For this purpose, we need to estimate a Hurst exponent at each date. We use an estimator exploiting the scaling properties of the fBm in a rolling window of size $T$, without regularizing the obtained time-varying Hurst exponent~\cite{Coeur2005,BP2010,Garcin2017}\footnote{ Precisely, in our empirical analysis, $T=504$ business days.}. More precisely, if we note $p_t$ the series of log-prices, supposed to follow an fBm of parameters $H$ and $\sigma$ in a specific sub-sample, any increment $p_t-p_{t-\tau}$ of duration $\tau>0$ contained in the sub-sample has a variance $\sigma^2|\tau|^{2H}$. Therefore, comparing the empirical log-variance of increments of two distinct durations, $\tau_1$ and $\tau_2$, leads to a simple estimator of the Hurst exponent using data in the interval $[t-T,t]$:
\begin{equation}\label{eq:estim}
\widehat H_t=\frac{1}{2\log(\tau_1/\tau_2)}\log\left(\frac{(T-\tau_2)\sum_{i=0}^{T-\tau_1}{(p_{t-i}-p_{t-i-\tau_1})^2}}{(T-\tau_1)\sum_{i=0}^{T-\tau_2}{(p_{t-i}-p_{t-i-\tau_2})^2}}\right).
\end{equation}

For the three daily series, the fBm-based forecast outperforms the forecast with the AR model, if we consider information criteria, such as AIC or BIC. However, the hit ratio for both methods is very close to $50\%$, whatever the lag. We assess the significance of the hit ratio with respect to $50\%$ through a binomial test. If we consider three one-day time lags, the hit ratio is $51.5\%$ (significantly above $50\%$ with a confidence of $98.3\%$) for the fBm and $51.8\%$ (conf. $99.4\%$) for the AR model, for the SPX series. For the CAC series, the hit ratio is $52.1\%$ (conf. $99.8\%$) for the fBm and $50.8\%$ (conf. $86.7\%$) for the AR model. For the GBPUSD series, the signal is not significantly different from noise, with a hit ratio of $50.2\%$ for the fBm and $50.1\%$ for the AR model. We will see in Section~\ref{sec:FX} how such low hit ratios can be improved by optimizing the time lags and thus by considering intraday data. However, some financial time series lead to high hit ratios even with a daily sampling. It is the case for example of the series of realized volatilities, which we study in the next section.

\subsubsection{Rough volatility}

The use of the fBm in finance is not limited to log-prices and one of the most popular applications of the fBm is for volatility modelling. Even though accuracy metrics directly related to the performance and risk of a trading strategy, as developed in Section~\ref{sec:seuilOptim}, are not appropriate in this case, forecasting a volatility is also useful in algorithmic trading, and one can properly use the hit ratio to evaluate the quality of the fBm-based forecast. If the idea of depicting the volatility with the fBm is not new~\cite{CR}, the recent paradigm is rough volatility, that is the fBm used for modelling volatility has a very low Hurst exponent, close to 0.15, so that the series of volatility is strongly antipersistent~\cite{ALV,GJR}. This model seems to currently dominate the research on stochastic volatility and leads to a rich literature, regarding for example some simplifications of the model to make it computationally more efficient~\cite{AbiJaber,AE}, or some refinements to make it more realistic~\cite{GG}. We also stress the fact that the rough framework relies on empirical observations of scaling rules of volatility increments and not on the analysis of the dependence between distant increments of volatility. This analysis of the long-range dependence seems to reveal a slowly decaying correlation~\cite{CV}, so that modelling volatility with a Hurst exponent higher than $1/2$ may also be relevant in certain cases~\cite{GS}.

The data used in our analysis are daily realized volatilities computed with a five-minute discretization of prices, imported from the Oxford-Man Institute of Quantitative Finance Realized Library
\footnote{ Available at \url{http://realized.oxford-man.ox.ac.uk/data/download}.}. We focus on the realized volatility of eight stock indices: the AEX index, the CAC 40 index, the FTSE 100 index, the Nasdaq 100 index (IXIC), the Nikkei 225 index (N225), the Oslo Exchange All-share index (OSEAX), the Madrid General index (SMSI), and the S\&P 500 index. The series starts on January 2000, except N225, which starts in February 2000, OSEAX in September 2001 and SMSI in July 2005. The end date of our sample is on the 12th April 2021.

Like for the series of prices above, we want to compare fBm-based predictions of log-volatilities with the AR model. The Hurst exponent, for the fBm approach, is estimated like in equation~\eqref{eq:estim}, in which the log-price $p_t$ is to be replaced by the logarithm of the realized volatility. We display the estimated time-varying Hurst exponent of the realized log-volatility of SPX in Figure~\ref{fig:HurstSPX}. We observe that, consistently with the literature on rough volatility, $\widehat H_t$ is mainly in the interval $[0.05,0.25]$.

\begin{figure}[htbp]
	\centering
		\includegraphics[width=0.75\textwidth]{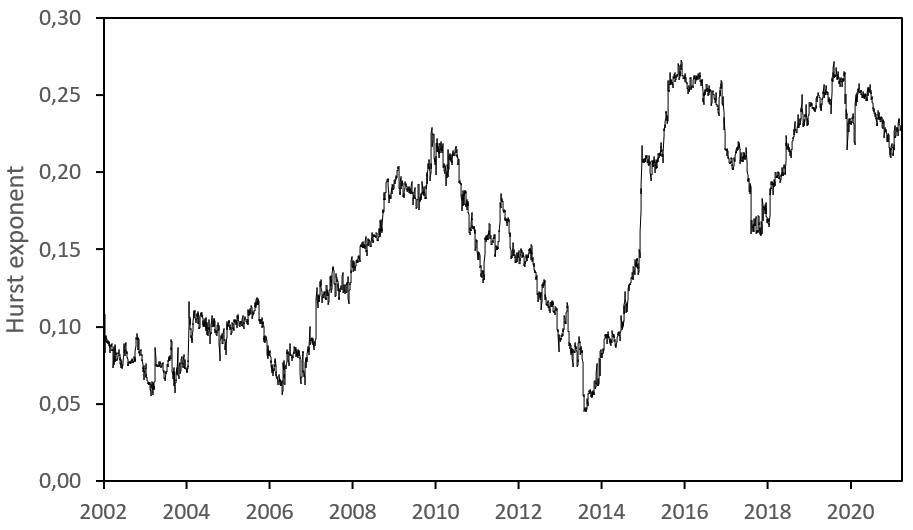} 
\begin{minipage}{0.7\textwidth}\caption{Estimated time-varying Hurst exponent of the realized log-volatility of SPX.}
	\label{fig:HurstSPX}
\end{minipage}
\end{figure}

Before presenting the accuracy of the forecast with the fBm approach for the log-volatilities series, we give more details about the benchmark alternative predictor. The alternative approach is an AR model with Gaussian residuals and with the same number $n$ of lagged daily increments of log-volatility as in the fBm predictor. The two models are estimated in the same rolling windows.  The fBm and the AR models are very close since both are Gaussian and make linear forecast. The only difference is that the weights are not calculated in the same way. For the AR model, each weight is a parameter. When increasing $n$, we add new parameters, which will maximize the in-sample accuracy. For the fBm approach, the $n$ weights all depend on the parameters $H$ and $\sigma$ of the fBm, so that increasing the number of lags will not increase the number of parameters. This last model is thus more parsimonious and less subject to overfitting than the AR approach. In other words, we expect the AR model to be more accurate in sample and the fBm to have better forecasting ability, at least when $n$ is big enough. The accuracy metric we present below will confirm this intuition.

We determine the likelihood of each model in sample. Logically, this likelihood, calculated on the estimation set, is higher for the AR model. But this is not a relevant indicator of forecasting ability. One often prefers using information criteria, such as AIC and BIC, which take into account, negatively, the log-likelihood and, positively, the number of parameters as well as, for the BIC, the number of data used for the estimation. Since the fBm model has less parameters than the AR model, it is less penalized and we can expect to have a lower AIC and BIC for it. 

Empirical results confirm the intuition in the sense that the fBm predictor has always a lower AIC when $n=6$, as one can see in Table~\ref{tab:AIC}\footnote{ The results with BIC, not displayed in a table, are similar.}. But the results are even stronger. Indeed, when $n=1$, that is when the variation of log-volatility is forecast only from the increment observed the last day, the fBm predictor has a slight advantage over the AR one, except for IXIC. We recall that the value of the information criterion cannot be interpreted in absolute value but that it is only useful for comparisons: for a same dataset, the best model is the one with the lowest criterion. For SPX, it seems from Figure~\ref{fig:SPX_accuracy}, that the optimal number of lagged observations for the AR predictor is $n=4$, at least regarding the BIC, whereas it is $n\geq 6$ for the fBm predictor. Indeed, adding more lagged observations does not increase the number of parameters only in the fBm predictor.

\begin{table}[htbp]
\centering
\begin{tabular}{l|cccc}
\hline
index & fBm(1) & fBm(6) & AR(1) & AR(6) \\
\hline
AEX & 875.9	& \textbf{812.9}	& 877.7	& 823.6 \\
CAC 40 & 795.0	& \textbf{734.8}	& 796.7	& 742.1 \\
FTSE 100 & 886.4	& \textbf{826.8}	& 888.0	& 837.7  \\
IXIC & 920.1	& \textbf{856.6}	& 916.8	& 861.2 \\
N225 & 770.9	& \textbf{716.8}	& 772.8	& 727.4 \\
OSEAX & 937.5	& \textbf{838.2}	& 939.1	& 845.5 \\
SMSI & 1025.4	& \textbf{931.1}	& 1027.4	& 940.6 \\
SPX & 933.6	& \textbf{873.9}	& 935.4	& 880.9 \\
\hline
\end{tabular}
\begin{minipage}{0.7\textwidth}\caption{AIC for the eight volatility series, for four different predictors: the fBm and the AR model both with 1 and 6 lagged observations. In bold is the lowest AIC for each series.}
\label{tab:AIC}
\end{minipage}
\end{table}

\begin{figure}[htbp]
	\centering
		\includegraphics[width=0.4\textwidth]{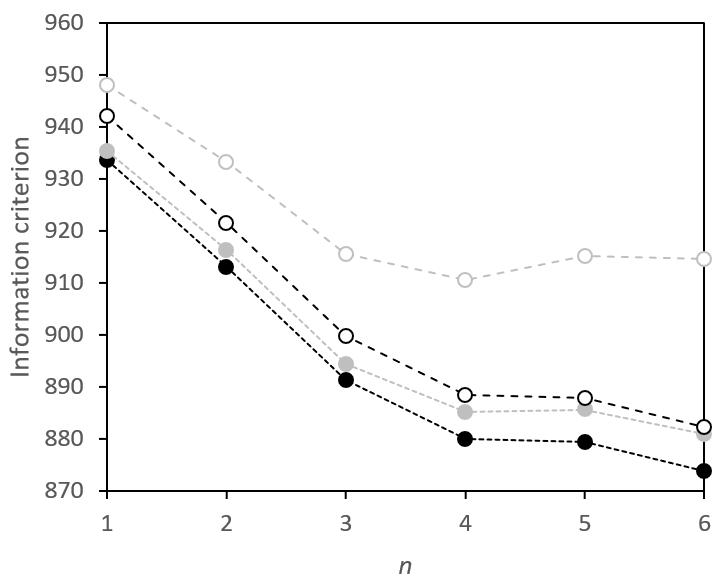}
		\includegraphics[width=0.4\textwidth]{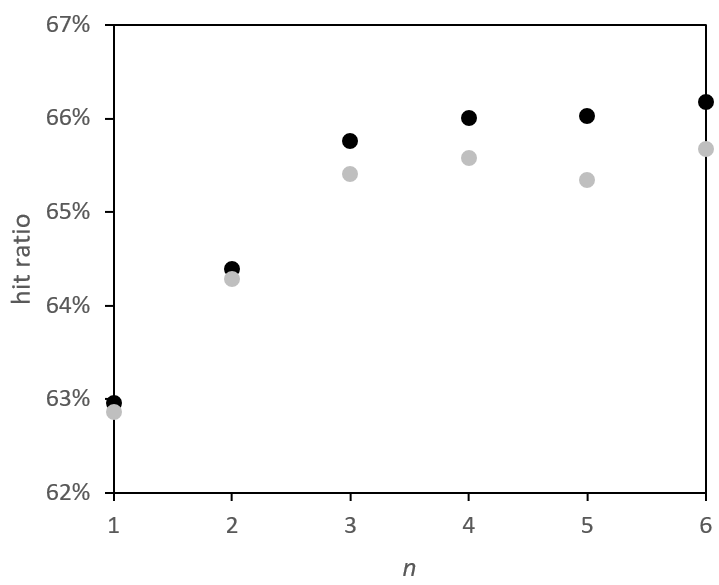} 
\begin{minipage}{0.7\textwidth}\caption{On the left: AIC (filled dots) and BIC (empty dots) for the fBm predictor (black) and the AR one (grey). On the right: out-of-sample hit ratio for the fBm predictor (black) and the AR one (grey). The number of lagged observations in the input of the two predictors is $n$ and the forecast series is the log-volatility of SPX.}
	\label{fig:SPX_accuracy}
\end{minipage}
\end{figure}

Information criteria provide in sample a hint of which model will have the higher out-of-sample accuracy. In order to check this intuition, we determine the empirical out-of-sample hit ratio for each model. Results for $n=1$ and $n=6$ are gathered in Table~\ref{tab:hit}. For $n=1$, the hit ratio is always higher for the fBm predictor than for the AR one, except for AEX index. For $n=6$, the fBm is always more accurate than any other model tested in the table. More details are displayed in Figure~\ref{fig:SPX_accuracy} for the SPX. It shows that the difference of accuracy between fBm and AR predictors becomes stronger from $n=3$. For both models, adding lagged observation has a strong impact on the accuracy until $n=3$. Our recommendation would thus be to use the fBm predictor with at least three lagged observations.

\begin{table}[htbp]
\centering
\begin{tabular}{l|cccc}
\hline
index & fBm(1) & fBm(6) & AR(1) & AR(6) \\
\hline
AEX & 62.4\% &	\textbf{65.8\%} &	62.5\% &	65.7\% \\
CAC 40 & 63.0\% &	\textbf{66.8\%} &	62.8\% &	66.0\% \\
FTSE 100 & 64.8\% &	\textbf{67.5\%} &	64.3\% &	67.4\% \\
IXIC & 62.2\% &	\textbf{64.6\%} &	62.1\% &	64.3\% \\
N225 & 62.9\% &	\textbf{67.4\%} &	62.8\% &	67.1\% \\
OSEAX & 64.0\% &	\textbf{67.9\%} &	63.8\% &	67.6\% \\
SMSI & 61.6\% &	\textbf{64.0\%} &	61.4\% &	63.7\% \\
SPX & 63.0\% &	\textbf{66.2\%} &	62.9\% &	65.7\% \\
\hline
\end{tabular}
\begin{minipage}{0.7\textwidth}\caption{Empirical out-of-sample hit ratio for the eight volatility series, for four different predictors: the fBm and the AR model both with 1 and 6 lagged observations. In bold is the highest hit ratio for each series.}
\label{tab:hit}
\end{minipage}
\end{table}

This analysis of realized volatility series shows that the fBm provides a good predictor, which is more accurate than an AR model because of its parsimony. It confirms previous findings in the literature, where the superiority of the fBm over the AR model was assessed with the help of another accuracy indicator, namely the MSE~\cite{GJR}. It highlights the relevance of the fBm in finance in a forecasting perspective.

\subsection{High-frequency FX rates}\label{sec:FX}

We now focus on time series of FX prices for three pairs: EURGBP, EURUSD, and GBPUSD. We consider high-frequency observations, with one price every minute. The FX market is continuously open during almost one week. We thus focus on one week in order to limit the effects of closing markets: from the 23rd June 2019 to the 28th June 2019.

We transform the price series in a series of log-prices and we assume it follows an fBm. We do not investigate here the relevance of the fBm for modelling FX log-prices, because many papers already deal with this question~\cite{BG,SGV,Garcin2017} and also propose some extensions of the fBm to match other stylized facts such as stationarity~\cite{GarcinLamperti}, high kurtosis~\cite{WBMW,GarcinMPRE,AG}, multiscaling~\cite{DiMatteo2007}, or even randomness of the Hurst exponent~\cite{BPP,GarcinMPRE}. We understand that several additional parameters may be relevant to perfectly depict the dynamic of FX log-prices, but, if our reference is the standard Brownian motion, the use of the fBm is a good step towards a realistic representation of this financial series.

Using a 12-hour window and equation~\eqref{eq:estim}, we estimate the Hurst exponent every minute for each series, as represented in Figure~\ref{fig:HurstFX}. Then, we forecast the evolution of the log-price at a one-hour horizon. We base this forecast on the formulas provided in Section~\ref{sec:forecast}. In particular, we take into account two lagged returns in each forecast. We propose either a naive choice of these lagged returns, or an \textit{optimal} one. The naive choice consists in the one-hour and two-hour past returns. Alternatively, the \textit{optimal} choice of lags is theoretically non-conditionally optimal\footnote{ This \textit{optimal} threshold is not to be confused with an ex-post optimal threshold.} and depends on the Hurst exponent, as exposed in Section~\ref{sec:optim}. For example, when $H=0.35$, the lagged returns considered are the last 21-minute and 176-minute returns.

\begin{figure}[htbp]
	\centering
		\includegraphics[width=0.75\textwidth]{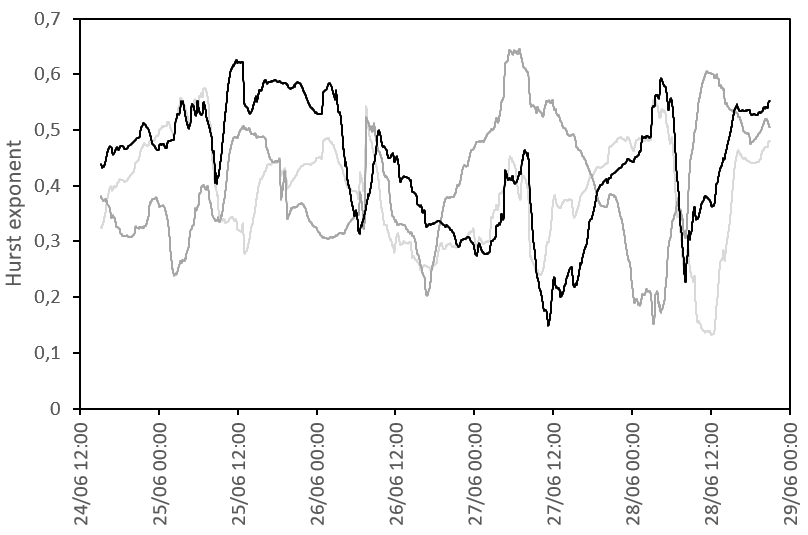} 
\begin{minipage}{0.7\textwidth}\caption{Estimated time-varying Hurst exponent of the log-prices of EURGBP (black), EURUSD (dark grey), GBPUSD (light grey).}
	\label{fig:HurstFX}
\end{minipage}
\end{figure}

As displayed in Table~\ref{tab:FX}, the hit ratios for the EURUSD pair are very close between the naive and the \textit{optimal} approaches. But for the two other pairs, the \textit{optimal} lags outperform the naive lags, for which the hit ratio is close to $50\%$. It thus sounds that using carefully the lags in the predictor can transform an equiprobable coin toss in a performing strategy. It is worth noting that increasing the number of lags, and not only choosing them properly, should also increase the hit ratio and make this predictor more appealing for practitioners.

\begin{table}[htbp]
\centering
\begin{tabular}{c|l|cccccc}
\hline
FX rate & Method & Return & Risk & Risk-adj. return & $\widehat{p^+}$ & $\widehat{p^-}$ & $\widehat{\ p^0\ }$ \\
\hline
 & Naive lags, $\theta=0$ & -0.51 & 2.47 & -2.97 & 49.8\% & 50.2\% & 0.0\% \\
EURGBP & \textit{Optimal} lags, $\theta=0$ & 3.66 & 2.26 & 1.40 & 52.9\% & 47.1\% & 0.0\% \\
 & \textit{Optimal} lags and $\theta$ & 3.09 & 1.50 & 1.59 & 30.2\% & 26.7\% & 43.1\% \\
\hline
 & Naive lags, $\theta=0$ & 2.82 & 1.85 & 0.97 & 54.6\% & 45.4\% & 0.0\% \\
EURUSD & \textit{Optimal} lags, $\theta=0$ & 1.52 & 1.91 & -0.39 & 54.3\% & 45.7\% & 0.0\% \\
 & \textit{Optimal} lags and $\theta$ & 1.37 & 1.20 & 0.17 & 34.4\% & 29.6\% & 36.0\% \\
\hline 
 & Naive lags, $\theta=0$ & 1.97 & 2.42 & -0.45 & 50.6\% & 49.4\% & 0.0\% \\
GBPUSD & \textit{Optimal} lags, $\theta=0$ & 4.22 & 2.31 & 1.91 & 52.4\% & 47.6\% & 0.0\% \\
 & \textit{Optimal} lags and $\theta$ & 4.58 & 1.63 & 2.96 & 33.5\% & 28.8\% & 37.7\% \\
\hline
\end{tabular}
\begin{minipage}{0.7\textwidth}\caption{Empirical average one-hour return, risk (lower absolute semi-deviation), and risk-adjusted return (with $\lambda=0.1$) multiplied respectively by $10^{5}$, $10^{4}$, and $10^{5}$. The three columns on the right are the corresponding empirical probabilities of good forecast, bad forecast, and zero forecast, following Definition~\ref{def:probaTernary}. As soon as $\theta=0$, $\widehat{p^+}$ is simply the hit ratio.}
\label{tab:FX}
\end{minipage}
\end{table}

We then implement the ternary strategy described in Section~\ref{sec:statarb}. We consider a threshold $\theta$ either equal to zero or to the theoretically optimal value determined numerically in Section~\ref{sec:seuilOptim} for a given risk aversion. Whatever the FX pair, we observe in Table~\ref{tab:FX} a sharp decrease of the risk when using the \textit{optimal} $\theta$. For EURGBP and GBPUSD, it also leads to the highest risk-adjusted performance. 

The EURUSD case is particular. We noticed the close hit ratios for the two kinds of lags, but the naive approach in fact leads to a better risk-adjusted performance. Three causes can explain this superiority of the naive approach for this FX pair. First, the \textit{optimal} choice of lags is intended to maximize a theoretical non-conditional hit ratio, not a risk-adjusted performance. Second, the non-conditional framework used to define optimal lags works in average, but in some cases a conditional approach may be preferable. Third and last, our formulas rely on a model, the fBm, which is only a simplification of reality, so that a wise practical implementation of the proposed predictor should constantly check the relevance of the model and challenge it with other approaches in order to minimize model risk. 

Nevertheless for the \textit{optimal} lags, we observe in Table~\ref{tab:FX}, that introducing the \textit{optimal} threshold instead of $\theta=0$ improves the risk-adjusted performance of the strategy for the three pairs, including EURUSD. This was precisely the purpose of the thresholding. We also note that the \textit{optimal} threshold leads not to take an investment decision roughly $40\%$ of the time.

\section{Conclusion}\label{sec:conclu}

In this paper, we have provided theoretical formulas for accuracy metrics related to forecasts of the fBm in discrete time. We focused on metrics that we consider to be meaningful for a systematic trader, namely the hit ratio, the average performance and risk. We have used these theoretical expressions to derive optimal time lags and an optimal threshold under which the prediction is to be discarded. Finally, we have demonstrated empirically the practical relevance of all these considerations on two kinds of time series: realized volatilities and high-frequency FX rates. We think that systematic traders should find in this contribution some elements for improving their strategies based on the fBm.

\bibliographystyle{plain}
\bibliography{biblioFracForec}

\begin{thebibliography}{10}

\bibitem{AbiJaber}
E.~Abi~Jaber.
\newblock Lifting the {H}eston model.
\newblock {\em Quantitative finance}, 19(12):1995--2013, 2019.

\bibitem{AE}
E.~Abi~Jaber and O.~El~Euch.
\newblock Multifactor approximation of rough volatility models.
\newblock {\em {SIAM} journal on financial mathematics}, 10(2):309--349, 2019.

\bibitem{ALV}
E.~Al\`os, J.A. Le\'on, and J.~Vives.
\newblock On the short-time behavior of the implied volatility for
  jump-diffusion models with stochastic volatility.
\newblock {\em Finance and stochastics}, 11(4):571--589., 2007.

\bibitem{ARAR}
J.~Alvarez-Ramirez, J.~Alvarez, E.~Rodriguez, and G.~Fernandez-Anaya.
\newblock Time-varying {H}urst exponent for {US} stock markets.
\newblock {\em Physica {A}: statistical mechanics and its applications},
  387(24):6159--6169, 2008.

\bibitem{AG}
A.~Ammy-Driss and M.~Garcin.
\newblock Efficiency of the financial markets during the {COVID}-19 crisis:
  time-varying parameters of fractional stable dynamics.
\newblock {\em arxiv preprint}, 2020.

\bibitem{BR}
F.M. Bandi and J.R. Russell.
\newblock Microstructure noise, realized variance, and optimal sampling.
\newblock {\em The review of economic studies}, 75(2):339--369, 2008.

\bibitem{BSV}
C.~Bender, T.~Sottinen, and E.~Valkeila.
\newblock Arbitrage with fractional {B}rownian motion?
\newblock {\em Theory of stochastic processes}, 13(1):23--34, 2007.

\bibitem{BP2010}
S.~Bianchi and A.~Pantanella.
\newblock Pointwise regularity exponents and market cross-correlations.
\newblock {\em International review of business research papers}, 6(2):39--51,
  2010.

\bibitem{BPP}
S.~Bianchi, A.~Pantanella, and A.~Pianese.
\newblock Modeling and simulation of currency exchange rates using
  multifractional process with random exponent.
\newblock {\em International journal of modeling and optimization},
  2(3):309--314, 2012.

\bibitem{BP}
S.~Bianchi and A.~Pianese.
\newblock Time-varying {H}urst-{H}\"older exponents and the dynamics of
  (in)efficiency in stock markets.
\newblock {\em Chaos, solitons \& fractals}, 109:64--75, 2018.

\bibitem{BG}
M.~Bohdalov\'a and M.~Gregu\v{s}.
\newblock Fractal analysis of forward exchange rates.
\newblock {\em Acta polytechnica hungarica}, 7(4):57--69, 2010.

\bibitem{CT}
D.O. Cajueiro and B.M. Tabak.
\newblock The {H}urst exponent over time: testing the assertion that emerging
  markets are becoming more efficient.
\newblock {\em Physica {A}: statistical mechanics and its applications},
  336(3-4):521--537, 2004.

\bibitem{Cheridito}
P.~Cheridito.
\newblock Arbitrage in fractional {B}rownian motion models.
\newblock {\em Finance and stochastics}, 7(4):533--553, 2003.

\bibitem{CV}
A.~Chronopoulou and F.G. Viens.
\newblock Estimation and pricing under long-memory stochastic volatility.
\newblock {\em Annals of finance}, 8(2-3):379--403, 2012.

\bibitem{Coeur2005}
J.-F. Coeurjolly.
\newblock Identification of multifractional {B}rownian motion.
\newblock {\em Bernoulli}, 11(6):987--1008, 2005.

\bibitem{CR}
F.~Comte and E.~Renault.
\newblock Long memory in continuous-time stochastic volatility models.
\newblock {\em Mathematical finance}, 8(4):291--323, 1998.

\bibitem{Cont}
R.~Cont.
\newblock Long range dependence in financial markets.
\newblock In J.~L\'evy-V\'ehel and E.~Lutton, editors, {\em Fractals in
  engineering}, pages 159--179. Springer, 2005.

\bibitem{DiMatteo2007}
T.~Di~Matteo.
\newblock Multi-scaling in finance.
\newblock {\em Quantitative finance}, 7(1):21--36, 2007.

\bibitem{DiMatteo2003}
T.~Di~Matteo, T.~Aste, and M.M. Dacorogna.
\newblock Scaling behaviors in differently developed markets.
\newblock {\em Physica {A}: statistical mechanics and its applications},
  324(1-2):183--188, 2003.

\bibitem{DiMatteo2005}
T.~Di~Matteo, T.~Aste, and M.M. Dacorogna.
\newblock Long-term memories of developed and emerging markets: {U}sing the
  scaling analysis to characterize their stage of development.
\newblock {\em Journal of banking \& finance}, 29(4):827--851, 2005.

\bibitem{ECOJ}
C.~Eom, S.~Choi, G.~Oh, and W.S. Jung.
\newblock Hurst exponent and prediction based on weak-form efficient market
  hypothesis of stock markets.
\newblock {\em Physica {A}: statistical mechanics and its applications},
  387(18):4630--4636, 2008.

\bibitem{Garcin2017}
M.~Garcin.
\newblock Estimation of time-dependent {H}urst exponents with variational
  smoothing and application to forecasting foreign exchange rates.
\newblock {\em Physica {A}: statistical mechanics and its applications},
  483:462--479, 2017.

\bibitem{GarcinLamperti}
M.~Garcin.
\newblock Hurst exponents and delampertized fractional {B}rownian motions.
\newblock {\em International journal of theoretical and applied finance},
  22(5):1950024, 2019.

\bibitem{GarcinEstimLamp}
M.~Garcin.
\newblock A comparison of maximum likelihood and absolute moments for the
  estimation of {H}urst exponents in a stationary framework.
\newblock {\em arxiv preprint}, 2020.

\bibitem{GarcinMPRE}
M.~Garcin.
\newblock Fractal analysis of the multifractality of foreign exchange rates.
\newblock {\em Mathematical methods in economics and finance}, 13-14(1):49--73,
  2020.

\bibitem{GG}
M.~Garcin and M.~Grasselli.
\newblock Long vs short time scales: the rough dilemma and beyond.
\newblock {\em arxiv preprint}, 2020.

\bibitem{GS}
J.~Garnier and K.~S{\o}lna.
\newblock Option pricing under fast-varying long-memory stochastic volatility.
\newblock {\em Mathematical finance}, 29(1):39--83, 2019.

\bibitem{GJR}
J.~Gatheral, T.~Jaisson, and M.~Rosenbaum.
\newblock Volatility is rough.
\newblock {\em Quantitative finance}, 18(6):933--949, 2018.

\bibitem{GSP}
M.S. Granero, J.T. Segovia, and J.G. P\'erez.
\newblock Some comments on {H}urst exponent and the long memory processes on
  capital markets.
\newblock {\em Physica {A}: statistical mechanics and its applications},
  387(22):5543--5551, 2008.

\bibitem{GMR}
P.~Guasoni, Y.~Mishura, and M.~R\'asonyi.
\newblock High-frequency trading with fractional {B}rownian motion.
\newblock {\em Finance and stochastics}, 25:277--310, 2021.

\bibitem{GNR}
P.~Guasoni, Z.~Nika, and M.~R\'asonyi.
\newblock Trading fractional {B}rownian motion.
\newblock {\em {SIAM} journal on financial mathematics}, 10(3):769--789, 2019.

\bibitem{Kristoufek}
L.~Kristoufek.
\newblock On {B}itcoin markets (in)efficiency and its evolution.
\newblock {\em Physica {A}: statistical mechanics and its applications},
  503:257--262, 2018.

\bibitem{KV13}
L.~Kristoufek and M.~Vosvrda.
\newblock Measuring capital market efficiency: {G}lobal and local correlations
  structure.
\newblock {\em Physica {A}: statistical mechanics and its applications},
  392(1):184--193, 2013.

\bibitem{KV16}
L.~Kristoufek and M.~Vosvrda.
\newblock Gold, currencies and market efficiency.
\newblock {\em Physica {A}: statistical mechanics and its applications},
  449:27--34, 2016.

\bibitem{MvN}
B.~Mandelbrot and J.~van Ness.
\newblock Fractional {B}rownian motions, fractional noises and applications.
\newblock {\em {SIAM} review}, 10(4):422--437, 1968.

\bibitem{NP}
C.J. Nuzman and H.V. Poor.
\newblock Linear estimation of self-similar processes via {L}amperti's
  transformation.
\newblock {\em Journal of applied probability}, 37(2):429--452, 2000.

\bibitem{Oomen}
R.C.A. Oomen.
\newblock Properties of realized variance under alternative sampling schemes.
\newblock {\em Journal of business \& economic statistics}, 24(2):219--237,
  2006.

\bibitem{Rogers}
L.C.G. Rogers.
\newblock Arbitrage with fractional {B}rownian motion.
\newblock {\em Mathematical finance}, 7(1):95--105, 1997.

\bibitem{SGV}
D.~Surgailis, G.~Teyssi\`ere, and M.~Vai\v{c}iulis.
\newblock The increment ratio statistic.
\newblock {\em Journal of multivariate analysis}, 99(3):510--541, 2008.

\bibitem{VanHarlow}
W.~Van~Harlow.
\newblock Asset allocation in a downside-risk framework.
\newblock {\em Financial analysts journal}, 47(5):28--40, 1991.

\bibitem{WBMW}
A.~Weron, K.~Burnecki, S.~Mercik, and K.~Weron.
\newblock Complete description of all self-similar models driven by {L}\'evy
  stable noise.
\newblock {\em Physical review {E}}, 71(1):016113, 2005.

\end{thebibliography}

\appendix

\section{Lemmas used for the proof of the theorems}\label{sec:lemma}

Let $g$ and $N$ be respectively the probability density function and the cumulative distribution function of a standard Gaussian variable. We introduce the following lemmas.

\begin{lem}\label{lem:hit}
For all $\alpha\in\mathbb R$, we have $\int_0^{\infty}{N(\alpha x)g(x)dx}=\frac{1}{4}+\frac{1}{2\pi}\arctan(\alpha)$.
\end{lem}

\begin{proof}
Let $f(\alpha)=\int_0^{\infty}{N(\alpha x)g(x)dx}$. Then $f$ is differentiable and
$$\begin{array}{ccl}
f'(\alpha) & = & \int_0^{\infty}{xg(\alpha x)g(x)dx} \\
 & = & \int_0^{\infty}{\frac{x}{2\pi}\exp\left(-\frac{x^2}{2}(1+\alpha^2)\right)dx} \\
 & = & \left[-\frac{1}{2\pi(1+\alpha^2)}\exp\left(-\frac{x^2}{2}(1+\alpha^2)\right)\right]^{\infty}_{x=0} \\
 & = & \frac{1}{2\pi(1+\alpha^2)}.
\end{array}$$
As a consequence, $f$ can be written in the form $f(\alpha)=\gamma+\frac{1}{2\pi}\arctan(\alpha)$, with $\gamma$ a constant. The constant $\gamma$ is equal to $f(0)=1/4$. This proves the lemma.
\end{proof}

\begin{lem}\label{lem:intBis}
For all $\alpha,a\in\mathbb R$, we have the following Taylor expansion in the neighbourhood of $a=0$:
$$\int_a^{\infty}{N(\alpha x)g(x)dx}=\frac{1}{4}+\frac{1}{2\pi}\arctan(\alpha)-\frac{a}{2\sqrt{2\pi}}-\frac{\alpha a^2}{4\pi}+\frac{a^3}{12\sqrt{2\pi}}+\frac{(\alpha^3+3\alpha)a^4}{48\pi}-\frac{a^5}{80\sqrt{2\pi}} +\mathcal O (a^6).$$
\end{lem}

\begin{proof}
Let $f(a)=\int_a^{\infty}{N(\alpha x)g(x)dx}$. First, getting the derivatives of $g$ is straightforward thanks to Hermite polynomials:
$$\left\{\begin{array}{ccl}
g(0) & = & \frac{1}{\sqrt{2\pi}} \\
g^{(1)}(0) & = & 0 \\
g^{(2)}(0) & = & - \frac{1}{\sqrt{2\pi}} \\
g^{(3)}(0) & = & 0 \\
g^{(4)}(0) & = &  \frac{3}{\sqrt{2\pi}} .
\end{array}\right.$$
The function $f$ is differentiable and, using Leibniz formula, we get:
\begin{enumerate}
\item $f^{(1)}(a)=-N(\alpha a)g(a)$ and $f^{(1)}(0)=-\frac{1}{2\sqrt{2\pi}}$,
\item $f^{(2)}(a)=-\alpha g(\alpha a)g(a)-N(\alpha a)g^{(1)}(a)$ and $f^{(2)}(0)=-\frac{\alpha}{2\pi}$,
\item $f^{(3)}(a)=-\alpha^2 g^{(1)}(\alpha a)g(a)-2\alpha g(\alpha a)g^{(1)}(a)-N(\alpha a)g^{(2)}(a)$ and $f^{(3)}(0)=\frac{1}{2\sqrt{2\pi}}$,
\item $f^{(4)}(a)=-\alpha^3 g^{(2)}(\alpha a)g(a)-3\alpha^2 g^{(1)}(\alpha a)g^{(1)}(a)-3\alpha g(\alpha a)g^{(2)}(a)-N(\alpha a)g^{(3)}(a)$ and $f^{(4)}(0)=\frac{\alpha^3+3\alpha}{2\pi}$,
\item $f^{(5)}(a)=-\alpha^4 g^{(3)}(\alpha a)g(a)-4\alpha^3 g^{(2)}(\alpha a)g^{(1)}(a)-6\alpha^2 g^{(1)}(\alpha a)g^{(2)}(a)-4\alpha g(\alpha a)g^{(3)}(a)-N(\alpha a)g^{(4)}(a)$ and $f^{(5)}(0)=-\frac{3}{2\sqrt{2\pi}}$.
\end{enumerate}
We conclude with Taylor's theorem applied to $f$ in $a=0$, using Lemma~\ref{lem:hit} for $f(0)$.
\end{proof}

\begin{lem}\label{lem:intTer}
For all $\alpha,a\in\mathbb R$, we have 
$$\int_a^{+\infty}{u N(\alpha u) g(u)du}=N(\alpha a)g(a)+\frac{\alpha}{\sqrt{2\pi(1+\alpha^2)}}\left(1-N\left(a\sqrt{1+\alpha^2}\right)\right).$$
\end{lem}

\begin{proof}
By noting that $-g$ is a primitive function of $u\mapsto ug(u)$, we integrate by part the following equation:
$$\begin{array}{ccl}
\int_a^{b}{u N(\alpha u) g(u)du} & = & \left[-N(\alpha u)g(u)\right]_a^b+\alpha\int_a^b{g(\alpha u)g(u)du} \\
& = & -N(\alpha b)g(b)+N(\alpha a)g(a)+\frac{\alpha}{\sqrt{2\pi}}\int_a^b{g\left(u\sqrt{1+\alpha^2}\right)du} \\
& = & -N(\alpha b)g(b)+N(\alpha a)g(a)+\frac{\alpha}{\sqrt{2\pi(1+\alpha^2)}}\left(N\left(b\sqrt{1+\alpha^2}\right)-N\left(a\sqrt{1+\alpha^2}\right)\right)
\end{array}$$
after a change of variable between the second and the third line. By doing $b\rightarrow+\infty$, we get Lemma~\ref{lem:intTer}.
\end{proof}

\section{Proof of Theorem~\ref{th:hitNrdt}}\label{sec:hitNrdt}

\begin{proof}
We treat separately the conditional and the non-conditional hit ratios. We begin with the non-conditional case.
\begin{enumerate}
\item According to equation~\eqref{eq:r_hat}, the predicted price return is the random variable
$$\hat R_{t,t+h|\Delta}=\Sigma_{RS}\Sigma_S^{-1}S.$$
Its mean is zero and its variance is:
$$\begin{array}{ccl}
\Var\left(\hat R_{t,t+h|\Delta}\right) & = & \E\left[\Sigma_{RS}\Sigma_S^{-1}S\left(\Sigma_{RS}\Sigma_S^{-1}S\right)^T\right] \\
& = & \Sigma_{RS}\Sigma_S^{-1}\E\left[SS^T\right]\left(\Sigma_S^T\right)^{-1}\Sigma_{RS}^T \\
& = & \Sigma_{RS}\Sigma_S^{-1}\Sigma_S\Sigma_S^{-1}\Sigma_{RS}^T \\
& = & \Sigma_{RS}\Sigma_S^{-1}\Sigma_{RS}^T ,
\end{array}$$
where we have used the fact that $\Sigma_S$ is symmetric. The covariance of $\hat R_{t,t+h|\Delta}$ with $R_{t,t+h}$ is:
$$\begin{array}{ccl}
\Cov(\hat R_{t,t+h|\Delta},R_{t,t+h}) & = & \E\left[\Sigma_{RS}\Sigma_S^{-1}SR_{t,t+h}\right] \\
& = & \Sigma_{RS}\Sigma_S^{-1}\E\left[SR_{t,t+h}\right] \\
& = & \Sigma_{RS}\Sigma_S^{-1}\Sigma_{RS}^T.
\end{array}$$
Then, the covariance matrix $\Gamma$ of the Gaussian vector $[\hat R_{t,t+h|\Delta} \ R_{t,t+h}]^T$ is:
$$\Gamma=\left(\begin{array}{cc}
 \Sigma_{RS}\Sigma_S^{-1}\Sigma_{RS}^T & \Sigma_{RS}\Sigma_S^{-1}\Sigma_{RS}^T \\
 \Sigma_{RS}\Sigma_S^{-1}\Sigma_{RS}^T & \sigma^2h^{2H}
\end{array}\right),$$
which we can also write $\Gamma=\Sigma\Sigma^T$, with $\Sigma$ a lower triangular matrix obtained by Cholesky decomposition:
$$\Sigma=\left(\begin{array}{cc}
a & 0 \\
 a & b
\end{array}\right),$$
with $a= \sqrt{\Sigma_{RS}\Sigma_S^{-1}\Sigma_{RS}^T}$ and $b=\sqrt{\sigma^2h^{2H} - \Sigma_{RS}\Sigma_S^{-1}\Sigma_{RS}^T}$. Therefore, there exists two Gaussian variables $U,V\sim\mathcal N (0,1)$, independent from each other, such that 
\begin{equation}\label{eq:CholUV}
[\hat R_{t,t+h|\Delta} \ R_{t,t+h}]^T=\Sigma [U\ V]^T.
\end{equation}
The non-conditional hit ratio can thus be written as a sum of joint probabilities of the two increments $\hat R_{t,t+h|\Delta}$ and $R_{t,t+h}$, which can be expressed as a combination of $U$ and $V$:
$$\begin{array}{ccl}
\rho & = & \proba\left[\hat R_{t,t+h|\Delta}\geq 0,R_{t,t+h}\geq 0\right]+\proba\left[\hat R_{t,t+h|\Delta}\leq 0,R_{t,t+h}\leq 0\right] \\
 & = & \proba\left[aU\geq 0,aU+bV\geq 0\right]+\proba\left[aU\leq 0,aU+bV\leq 0\right] \\
 & = & \proba\left[U\geq 0,  V\geq -(a/b) U\right] + \proba\left[U\leq 0,  V\leq -(a/b) U\right] \\
 & = & \int_{0}^{+\infty}{\int_{-au/b}^{+\infty}{(g(u)g(v))dv}du} + \int_{-\infty}^{0}{\int_{-\infty}^{-au/b}{(g(u)g(v))dv}du} \\
 & = & \int_{0}^{+\infty}{\left(1-N\left(-\frac{a}{b}u\right)\right)g(u)du} + \int_{-\infty}^{0}{N\left(-\frac{a}{b}u\right)g(u)du} \\
 & = & \frac{1}{2} - \int_{0}^{+\infty}{N\left(-\frac{a}{b}u\right)g(u)du} + \int_{0}^{+\infty}{N\left(\frac{a}{b}u\right)g(u)du} \\
 & = & \frac{1}{2} - \frac{1}{4}-\frac{1}{2\pi}\arctan\left(-\frac{a}{b}\right) + \frac{1}{4}+\frac{1}{2\pi}\arctan\left(\frac{a}{b}\right) \\
 & = & \frac{1}{2}+\frac{1}{\pi}\arctan\left(\frac{a}{b}\right).
\end{array}$$
where $g$ and $N$ are respectively the probability density function and the cumulative distribution function of a standard Gaussian variable and where we used Lemma~\ref{lem:hit}, which is in Appendix~\ref{sec:lemma}. Moreover, we note that
$$\frac{a}{b} =  \frac{\sqrt{\Sigma_{RS}\Sigma_S^{-1}\Sigma_{RS}^T}}{\sqrt{\sigma^2h^{2H} - \Sigma_{RS}\Sigma_S^{-1}\Sigma_{RS}^T}} = \left(\frac{\sigma^2h^{2H}}{\Sigma_{RS}\Sigma_S^{-1}\Sigma_{RS}^T}-1\right)^{-1/2},$$
so that
$$\rho = \frac{1}{2}+\frac{1}{\pi}\arctan\left(\left[\frac{\sigma^2h^{2H}}{\Sigma_{RS}\Sigma_S^{-1}\Sigma_{RS}^T}-1\right]^{-1/2}\right) = 1-\frac{1}{\pi}\arctan\left(\sqrt{\frac{\sigma^2h^{2H}}{\Sigma_{RS}\Sigma_S^{-1}\Sigma_{RS}^T}-1}\right).$$
\item We now deal with the conditional case. Thanks to the above Cholesky decomposition and to equation~\eqref{eq:CholUV}, we know that, conditionally to $S$, the variable $R_{t,t+h}$ is Gaussian of mean $\E\{R_{t,t+h}|S\}=\hat R_{t,t+h|\Delta}$ and of variance $b^2$.  Therefore, the conditional hit ratio is, where $V\sim\mathcal N (0,1)$:
$$\begin{array}{ccl}
\rho^c(y) & = & \proba\left[\hat R_{t,t+h|\Delta}\geq 0,\hat R_{t,t+h|\Delta}+bV\geq 0|S=y\right] + \proba\left[\hat R_{t,t+h|\Delta}< 0,\hat R_{t,t+h|\Delta}+bV< 0|S=y\right] \\
 & = & \indic_{\{\Sigma_{RS}\Sigma_S^{-1}y\geq 0\}} \proba\left[V\geq -\Sigma_{RS}\Sigma_S^{-1}y/b\right] + \indic_{\{\Sigma_{RS}\Sigma_S^{-1}y< 0\}} \proba\left[V< -\Sigma_{RS}\Sigma_S^{-1}y/b\right] \\
  & = & \indic_{\{\Sigma_{RS}\Sigma_S^{-1}y\geq 0\}} N\left(\Sigma_{RS}\Sigma_S^{-1}y/b\right) + \indic_{\{\Sigma_{RS}\Sigma_S^{-1}y< 0\}} N\left(-\Sigma_{RS}\Sigma_S^{-1}y/b\right) \\
  & = & N\left(\left|\Sigma_{RS}\Sigma_S^{-1}y\right|/b\right).
\end{array}$$
\end{enumerate}
\end{proof}

\section{Proof of Theorem~\ref{th:ternary}}\label{sec:ternary}

\begin{proof}
Using the same approach as in the proof of Theorem~\ref{th:hitNrdt}, we can write:
\begin{equation}\label{eq:p+}
\begin{array}{ccl}
p^+(\theta) & = & \proba\left[\hat R_{t,t+h|\Delta}\geq \theta,R_{t,t+h}> 0\right] +\proba\left[\hat R_{t,t+h|\Delta}\leq -\theta,R_{t,t+h}< 0\right] \\
 & = & \proba\left[aU\geq\theta,aU+bV>0\right] + \proba\left[aU\leq-\theta,aU+bV<0\right] \\
 & = & \proba\left[U\geq \theta/a,  V> -(a/b)U\right] + \proba\left[U\leq -\theta/a,  V< -(a/b) U\right] \\
 & = & \int_{\theta/a}^{+\infty}{\int_{-au/b}^{+\infty}{(g(u)g(v))dv}du} + \int_{-\infty}^{-\theta/a}{\int_{-\infty}^{-au/b}{(g(u)g(v))dv}du} \\
 & = & \int_{\theta/a}^{+\infty}{\left(1-N\left(-\frac{a}{b}u\right)\right)g(u)du} + \int_{-\infty}^{-\theta/a}{N\left(-\frac{a}{b}u\right)g(u)du} \\
 & = & 1-N(\theta/a)-\int_{\theta/a}^{+\infty}{N(-\frac{a}{b}u)g(u)du}+\int_{\theta/a}^{+\infty}{N(\frac{a}{b} u)g(u)du},
\end{array}
\end{equation}
by symmetry of $g$, the Gaussian density, with a change of variable in the second integral, and where $U,V\sim\mathcal N (0,1)$ are random variables independent from each other, $a= \sqrt{\Sigma_{RS}\Sigma_S^{-1}\Sigma_{RS}^T}$, and $b=\sqrt{\sigma^2h^{2H} - \Sigma_{RS}\Sigma_S^{-1}\Sigma_{RS}^T}$. Using Lemma~\ref{lem:intBis}, which is in Appendix~\ref{sec:lemma}, we obtain:
$$p^+(\theta)=1-N(\theta/a)+\frac{1}{\pi}\arctan(\frac{a}{b})-\frac{\theta^2}{2\pi ab}+\left(\frac{1}{ab^3}+\frac{3}{a^4b}\right)\frac{\theta^4}{24\pi}+\mathcal O (\theta^6).$$
Similarly, we can write
\begin{equation}\label{eq:p-}
\begin{array}{ccl}
p^-(\theta) & = & \proba\left[aU\leq-\theta,aU+bV>0\right] + \proba\left[aU\geq\theta,aU+bV<0\right] \\
 & = & \int_{\theta/a}^{+\infty}{N(-\frac{a}{b} u)g(u)du}+\int_{-\infty}^{-\theta/a}{(1-N(-\frac{a}{b} u))g(u)du} \\
 & = & 1-N(\theta/a)+\int_{\theta/a}^{+\infty}{N(-\frac{a}{b} u)g(u)du}-\int_{\theta/a}^{+\infty}{N(\frac{a}{b} u)g(u)du}. 
\end{array}
\end{equation}
Using again Lemma~\ref{lem:intBis}, we get:
$$p^-(\theta)=1-N(\theta/a)-\frac{1}{\pi}\arctan(\frac{a}{b})+\frac{\theta^2}{2\pi ab}-\left(\frac{1}{ab^3}+\frac{3}{a^4b}\right)\frac{\theta^4}{24\pi}+\mathcal O (\theta^6).$$
From equations~\eqref{eq:p+} and~\eqref{eq:p-}, we also obtain:
$$p^0(\theta) = 1- p^+(\theta)-p^-(\theta)=-1+2N\left(\frac{\theta}{a}\right).$$
This proves Theorem~\ref{th:ternary}.
\end{proof}

\section{Proof of Theorem~\ref{th:ternaryFin}}\label{sec:ternaryFin}

\begin{proof}
By definition,
$$\widetilde R_{t,t+h}(\theta)  = \E\left[R_{t,t+h}\left(\indic_{\{\hat R_{t,t+h|\Delta}\geq \theta\}} -\indic_{\{\hat R_{t,t+h|\Delta}\leq -\theta\}}\right)\right].$$
Like in the proof of Theorems~\ref{th:hitNrdt} and~\ref{th:ternary}, we can write it:
$$\widetilde R_{t,t+h}(\theta)  = \E\left[\left(aU+bV\right)\left(\indic_{\{a U\geq \theta\}} -\indic_{\{a U\leq -\theta\}}\right)\right],$$
where $U,V\sim\mathcal N (0,1)$ are two independent Gaussian variables, $a= \sqrt{\Sigma_{RS}\Sigma_S^{-1}\Sigma_{RS}^T}$, and $b=\sqrt{\sigma^2h^{2H} - \Sigma_{RS}\Sigma_S^{-1}\Sigma_{RS}^T}$. Then
$$\begin{array}{ccl}
\widetilde R_{t,t+h}(\theta) & = & \int_{\theta/a}^{+\infty}{\left\{\int_{\mathbb R}{\left(au+bv\right)g(v)g(u)dv}\right\}du} - \int_{-\infty}^{-\theta/a}{\left\{\int_{\mathbb R}{\left(au+bv\right)g(v)g(u)dv}\right\}du} \\
& = & \int_{\theta/a}^{+\infty}{a ug(u)du} - \int_{-\infty}^{-\theta/a}{a ug(u)du} \\
& = & 2ag(\theta/a),
\end{array}$$
using the symmetry of $g$ and the fact that $ug(u)=-g'(u)$. Similarly,
$$\begin{array}{ccl}
\widetilde\sigma_{t,t+h}^-(\theta) & = & -\E\left[R_{t,t+h}\left(\indic_{\{\hat R_{t,t+h|\Delta}\geq \theta,R_{t,t+h}<0\}} -\indic_{\{\hat R_{t,t+h|\Delta}\leq -\theta,R_{t,t+h}>0\}}\right)\right] \\
 & = & -\E\left[\left(aU+bV\right)\left(\indic_{\{a U\geq \theta,aU+bV<0\}} -\indic_{\{aU\leq -\theta,aU+bV>0\}}\right)\right] \\
 & = & -\int_{\theta/a}^{+\infty}{\left\{\int_{-\infty}^{-au/b}{\left(au+bv\right)g(v)g(u)dv}\right\}du} \\
 & & + \int_{-\infty}^{-\theta/a}{\left\{\int_{-au/b}^{+\infty}{\left(au+bv\right)g(v)g(u)dv}\right\}du} \\
 & = & -\int_{\theta/a}^{+\infty}{\left[au N\left(-\frac{au}{b}\right) - b g\left(-\frac{au}{b}\right)\right]g(u)du} \\
  & & + \int_{-\infty}^{-\theta/a}{\left[au \left(1-N\left(-\frac{au}{b}\right)\right) + b g\left(-\frac{au}{b}\right)\right]g(u)du} \\
  & = & -2a\int_{\theta/a}^{+\infty}{u N\left(-\frac{au}{b}\right)g(u)du} + 2b\int_{\theta/a}^{+\infty}{ g\left(\frac{au}{b}\right)g(u)du},
\end{array}$$
after a change of variable $u\mapsto -u$ in the last line. This leads to the general formula, where we use Lemma~\ref{lem:intTer} (which is postponed in Appendix~\ref{sec:lemma}) in the first integral, the fact that $g(au/b)g(u)=g(u\sqrt{1+(a/b)^2})/\sqrt{2\pi}$ and a change of variable $u\mapsto u\sqrt{1+(a/b)^2}$ in the second integral:
$$\begin{array}{ccl}
\widetilde\sigma_{t,t+h}^-(\theta)  & = & - 2a \left[N\left(-\frac{\theta}{b}\right)g\left(\frac{\theta}{a}\right)-\frac{a}{\sqrt{2\pi(a^2+b^2)}}\left(1-N\left(\theta\sqrt{\frac{1}{a^2}+\frac{1}{b^2}}\right)\right)\right]  \\
 & & +  \frac{2b}{\sqrt{2\pi}\sqrt{1+(a/b)^2}}\left(1-N\left(\theta\sqrt{\frac{1}{a^2}+\frac{1}{b^2}}\right)\right) \\
  & = & - 2aN\left(-\frac{\theta}{b}\right)g\left(\frac{\theta}{a}\right) + \frac{2a^2+2b^2}{\sqrt{2\pi(a^2+b^2)}}\left(1-N\left(\theta\sqrt{\frac{1}{a^2}+\frac{1}{b^2}}\right)\right) \\
  & = & - 2aN\left(-\frac{\theta}{b}\right)g\left(\frac{\theta}{a}\right) + \sqrt{\frac{2(a^2+b^2)}{\pi}}N\left(-\theta\sqrt{\frac{1}{a^2}+\frac{1}{b^2}}\right).
\end{array}$$
By noting that $a^2+b^2=\sigma^2 h^{2H}$, we obtain the results displayed in Theorem~\ref{th:ternaryFin}.
\end{proof}

\end{document}